\pdfoutput=1
\documentclass[11pt]{article}

\usepackage{microtype} %
\usepackage{hyperref}

\usepackage{url}
\usepackage{cite}
\usepackage{graphicx}
\usepackage{mathtools} %
\usepackage{amsthm}
\usepackage{amsfonts}
\usepackage{amssymb}
\usepackage{fullpage}
\usepackage{color}
\usepackage{latexsym}
\usepackage{enumerate}
\usepackage{fixltx2e}

\newtheorem{theorem}{Theorem}[section]
\newtheorem{corollary}[theorem]{Corollary}
\newtheorem{lemma}[theorem]{Lemma}

\newtheorem{fact}[theorem]{Fact}
\theoremstyle{remark}
\newtheorem*{note}{Note}
\newtheorem*{notes}{Notes}

\newcommand{\old}[1]{{{}}}

\DeclareMathOperator{\area}{\mathit{area}}
\newcommand{\Reals}{{\mathbb R}}

\def\L{\mathcal{L}}
\def\Q{\mathcal{Q}}
\def\S{\mathcal{S}}
\def\RR{\mathbb{R}}
\def\ff{\tilde{f}}
\def\EE{\tilde{E}}

\let\eps\varepsilon

\begin{document}

\title{Batched Point Location in SINR Diagrams\\via Algebraic Tools\thanks{%
    An earlier version of this paper (excluding Sections~\ref{sec:structural} and~\ref{sec:speed-up}) was presented at ICALP'15
    \cite{ak-bpsdat-15}; see arXiv \cite{this-arXiv} for a more complete version.
    Work on this paper by B.A.\ has been partially supported by
		NSF Grants CCF-11-17336, CCF-12-18791, and CCF-15-40656, and by grant 2014/170 from the US-Israel Binational Science Foundation.
    Work on this paper by M.K.\ has been supported by grant 2014/170 from the US-Israel Binational Science Foundation, and by
    grants 1045/10 and 1884/16 from the Israel Science Foundation. }}

  \author{Boris Aronov%
    \thanks{%
      Department of Computer Science and Engineering, Tandon
      School of Engineering, New York University, Brooklyn, NY~11201,
      USA; \texttt{boris.aronov@nyu.edu}.}
    \and
    Matthew J. Katz%
    \thanks{%
      Department of Computer Science, Ben-Gurion University, Israel;
      \texttt{matya@cs.bgu.ac.il}.}  }
\maketitle

\begin{abstract} 
The \emph{SINR model} for the quality of wireless connections has been the subject of extensive recent study.  It attempts to predict whether a particular transmitter is heard at a specific location, in a setting consisting of $n$~simultaneous transmitters and background noise.  The SINR model gives rise to a natural geometric object, the \emph{SINR diagram}, which partitions the space into $n$ regions where each of the transmitters can be heard and the remaining space where no transmitter can be heard.

  Efficient \emph{point location} in the SINR diagram, i.e., being able to build a data structure that facilitates determining, for a query point, whether any transmitter is heard there, and if so, which one, has been recently investigated in several papers.  These planar data structures are constructed in time at least quadratic in~$n$ and support logarithmic-time approximate queries. Moreover, the performance of some of the proposed structures depends strongly not only on the number~$n$ of transmitters and on the approximation parameter $\eps$, but also on some geometric parameters that cannot be bounded \emph{a priori} as a function of~$n$ or~$\eps$.

  In this paper, we address the question of \emph{batched} point location queries, i.e., answering many queries simultaneously.  Specifically, in one dimension, we can answer $n$~queries \emph{exactly} in amortized polylogarithmic time per query, while in the plane we can do it approximately.

  In another result, we show how to answer $n^2$~queries \emph{exactly} in amortized polylogarithmic time per query, assuming the queries are located on a possibly non-uniform $n \times n$ grid. 

  All these results can handle \emph{arbitrary} power assignments to the transmitters. Moreover, the amortized query time in these results depends only on $n$ and $\eps$.

  We also show how to speed up the preprocessing in a previously proposed point-location structure in SINR diagram for uniform-power sites, by almost a full order of magnitude. For this we obtain results on the sensitivity of the reception regions to slight changes in the reception threshold, which are of independent interest. 
		
  Finally, these results demonstrate the (so far underutilized) power of combining algebraic tools with those of computational geometry and other fields.
\end{abstract}

\thispagestyle{empty}
\newpage
\setcounter{page}{1}

\section{Introduction}

The \emph{SINR (Signal to Interference plus Noise Ratio) model} attempts to more realistically predict whether a wireless transmission is received successfully, in a setting consisting of multiple simultaneous transmitters in the presence of background noise. In particular, it takes into account the attenuation of electromagnetic signals. The SINR model has been explored extensively in the literature; see \cite{lp-saSINRwm-10} for a survey. 

Let $\S = \{ s_1, \dots, s_n\}$ be a set of $n$ points in the plane representing the locations of $n$ transmitters.
Let $p_i>0$ be the transmission power of transmitter $s_i$, $i = 1, \ldots, n$. 
In the \emph{SINR model}, a receiver located at point $q$ is able to receive the signal transmitted by $s_i$ if the following inequality holds:
\[
\frac{\frac{p_i}{|q-s_i|^\alpha}}{\Sigma_{j \ne i}{\frac{p_j}{|q-s_j|^\alpha}} + N} \geq \beta \, ,
\]
where $|a-b|$ denotes the Euclidean distance between points $a$ and
$b$, and $\alpha > 0$, $\beta > 1$,%
and $N > 0$ are given constants; $N$ represents the background noise.
This inequality is also called the \emph{SINR inequality}, and when it holds, we say that $q$ \emph{receives} (or \emph{hears}) $s_i$; we refer to the left hand side of the inequality as \emph{SIN~ratio} (for receiver~$q$ with respect to transmitter $s_i$) and to the right hand side as the \emph{reception threshold}.

Notice that, since $\beta > 1$, a necessary condition for $q$ to receive $s_i$ is that $p_i/|q-s_i|^\alpha > p_j/|q-s_j|^\alpha$, for any~$j \ne i$. In particular, in the \emph{uniform power setting} where $p_1=p_2=\dots=p_n$, a necessary condition for $q$ to receive $s_i$ is that $s_i$ is the closest to $q$ among the transmitters in $\S$. This simple observation implies that, for any point $q$ in the plane, either exactly one of the transmitters is received by $q$ or none of them is. Thus, one can partition the plane into $n$ not necessarily connected reception regions~$R_i$, one per transmitter in $\S$, plus an additional region~$R_\emptyset$ consisting of all points where none of the transmitters is received. This partition is called the \emph{SINR diagram} of~$\S$ \cite{aeklpr-sdciawn-12}.
Consider the \emph{multiplicatively weighted Voronoi diagram} $D$ of $\S$ in which the region $V_i$ associated with $s_i$ consists of all points $q$ in the plane for which $\frac{1}{\sqrt[\alpha]{p_i}}|q-s_i| < \frac{1}{\sqrt[\alpha]{p_j}}|q-s_j|$, for any $j \ne i$ (see Section~\ref{sec:voronoi}).
Then $R_i \subset V_i$.

In a seminal paper, Avin et  al. \cite{aeklpr-sdciawn-12} studied properties of SINR diagrams, focusing on the uniform power setting. Their main result is that in this setting the reception regions~$R_i$ are convex and fat. (Here, $R_i$ is \emph{fat} if the ratio between the radii of the smallest disk centered at $s_i$ containing $R_i$ and the largest disk centered at $s_i$ contained in $R_i$ is bounded by some constant.) In the non-uniform power setting, on the other hand, the reception regions are not necessarily connected, and their connected components are not necessarily convex or fat.  In fact, they may contain holes \cite{klpp-twn-11}.

A natural question that one may ask is: ``Given a point $q$ in the plane, does $q$ receive one of the transmitters in $\S$, and if yes which one?'' 
Or equivalently: ``Which region of the SINR diagram does $q$ belong to?'' The latter question is referred to as a \emph{point-location query} in the SINR diagram of $\S$. 
We can answer it in linear time by first finding the sole candidate, $s_i$, as the transmitter for which the ratio $\frac{1}{\sqrt[\alpha]{p}}|q-s|$ is minimum, and then evaluating the SIN~ratio and comparing it to~$\beta$.  To facilitate multiple queries, one may want to build a data structure that can guarantee faster response.
We can expedite the first step by constructing the appropriate Voronoi diagram $D=D(\S)$ 
together with a point-location structure, so that the sole candidate transmitter for a point $q$ can be found in $O(\log n)$ time; see, e.g.,\cite{AKL-VD}.
However, the boundary of the region $R_i$ is described by a degree-$\Theta(n)$ algebraic curve; it seems difficult (impossible, in general?) to build a data structure that can quickly determine the side of the curve a given point lies on.  The answer is not even obvious in one dimension (where the transmitters and potential receivers all lie on a line), as there $R_i$ is a collection of intervals delimited by roots of a polynomial of degree $\Theta(n)$. 

The problem has been approached by constructing data structures for \emph{approximate} point location in SINR diagrams.  All approaches use essentially the same logic: first find the sole candidate $s_i$ that the query point~$q$ may hear and then approximately locate $q$ in $R_i$.  This is done by constructing two sets $R_i^+$, $R_i^-$ such that $R_i^+ \subset R_i \subset R_i^- \subset V_i$,\footnote{%
   Notice that we have not followed the original notation in the literature, for consistency with our notation below.}
and preprocessing them for point location.
In the region $R_i^+$ reception of $s_i$ is guaranteed, so if $q \in R_i^+$, return ``can hear $s_i$.''  Outside of $R_i^-$ one cannot hear $s_i$, so if $q \in V_i$ does not belong to $R_i^-$, return ``cannot hear anything.''  The set $R_i^- \setminus R_i^+$ is where the approximation occurs: $s_i$ may or may not be heard there, so if $q \in R_i^-$ but $q \not \in R_i^+$, return ``may or may not hear $s_i$.''

Two different notions of approximation have appeared in the literature. In the first \cite{aeklpr-sdciawn-12,klpp-twn-11}, it is guaranteed that the uncertain answer is only given infrequently, namely that $\area(R_i^- \setminus R_i^+) \leq \eps \cdot \area (R_i)$, for a prespecified parameter $\eps>0$.  In the second \cite{klpp-twn-11}, it is promised that the SIN~ratio for every point in $R_i^- \setminus R_i^+$ lies within $[C_1 \beta,C_2 \beta]$ for suitable constants $C_1$, $C_2$ with $0<C_1<1$, $C_2>1$.
We show below (see end of Section~\ref{sec:structural}) how the two notions are related, at least for the uniform power case.

We now briefly summarize previous work. %
Observing the difficulty of answering point-location queries exactly, Avin et al. \cite{aeklpr-sdciawn-12} resorted to approximate query answers in the \emph{uniform power} setting.  Given an $\eps>0$ they build a data structure in total time\footnote{%
See below for discussion of an apparent discrepancy between the claimed preprocessing time of their data structure and the analysis presented in their paper; the actual preprocessing time seems to be $O(n^2(\log n + 1/\eps))$.}
$O(n^2/\eps)$ and space $O(n/\eps)$ that can be wrong only in a region of area $\eps\cdot\area(R_i)$ for each $s_i$ (i.e., approximation of the first type described above).  It supports logarithmic-time queries.

In a subsequent paper, Kantor et al. \cite{klpp-twn-11} studied properties of SINR diagrams in the \emph{non-uniform power} setting. After revealing several interesting and useful properties, such as that the reception regions in the ($d+1$)-dimensional SINR diagram of a $d$-dimensional scene are connected, they present several solutions to the problem of efficiently answering point-location queries.  One of them uses the second type of approximation, with $C_1=(1-\eps)^{2\alpha}$ and $C_2=(1+\eps)^{2\alpha}$, for a prespecified $\eps>0$.  Queries can be performed in time $O(\log(n \cdot \varphi / \eps))$, where $\varphi$ is an upper bound on the fatness parameters of the reception regions (which cannot be bounded as a function of~$n$ or~$\eps$). The size of this data structure is $O(n \cdot \varphi' / \eps^2)$ and its construction time is $O(n^2 \cdot \varphi' / \eps^2)$, where $\varphi' > \varphi^2$ is some function of the fatness parameters of the reception regions.

Although highly non-trivial, the known results for point location in the SINR model are unsatisfactory, in that they suffer from very large preprocessing times. Moreover, in the non-uniform setting, the bounds include geometric parameters such as $\varphi$ and $\varphi'$ above, which cannot be bounded as a function of~$n$ or~$\eps$.
In this paper we focus on \emph{batched} point location in the SINR  model. 
That is, given a set $\Q$ of $m$ query points, determine for each point $q \in \Q$ whether it receives one of the transmitters in $\S$,
and if yes, which one. Often the set of query points is known in advance, for example, 
in the planning stage of a wireless network or when examining an existing network.  In these cases, one would like to exploit the additional information to speed up query processing.  We achieve this goal in the SINR model; that is, we devise efficient approximation and exact algorithms for batched point location in various settings. Our algorithms use a novel combination of 
geometric data structures and tools from computer algebra for multipoint evaluation, interpolation, and fast multiplication of polynomials and rational functions.
For example, consider 1-dimensional batched point location where $m=n$ and power is non-uniform. We can answer \emph{exactly} a point-location query in amortized time $O(\log^2 n \log \log n)$.  Considering the same problem in the plane, for any $\eps > 0$, we can approximately answer a query in amortized time polylogarithmic in $n$ and $\eps$, as opposed to the result of Kantor et al. \cite{klpp-twn-11} mentioned above in which the bounds depend on additional geometric parameters which cannot be bounded as a function of~$n$ or~$\eps$.  See Section~\ref{sec:results} for a detailed list of our results.

\subsection{Related work}
The papers most relevant to ours are those by Avin et al. \cite{aeklpr-sdciawn-12} and Kantor et al. \cite{klpp-twn-11} discussed above. Avin et al. \cite{SINR-Diag-SIC} also considered the problem of handling queries of the following form (in the uniform-power setting): Given a transmitter $s_i$ and query point $q$, does $q$ receive $s_i$ by successively applying interference cancellation? (Interference cancellation is a technology that enables a point $q$ to receive a transmitter $s$, even if $s$'s signal is not the strongest one received at $q$; see \cite{SINR-Diag-SIC} for further details.)

Gupta and Kumar \cite{GK00} initiated an extensive study of the \emph{maximum capacity} and \emph{scheduling} problems in the SINR model.
Given a set $L$ of sender-receiver pairs (i.e., directional links), the \emph{maximum capacity} problem is to find a \emph{feasible} subset of $L$ of maximum cardinality, where $L' \subseteq L$ is \emph{feasible} if, when only the senders of the links in $L'$ are active, each of the links in $L'$ is feasible according to the SINR inequality. The \emph{scheduling} problem is to partition $L$ into a minimum number of feasible subsets (i.e., rounds).  
We mention several papers and results dealing with the maximum capacity and scheduling problems. 
Goussevskaia et~al. \cite{GOW07} showed that both problems are NP-complete, even in the uniform power setting.  
Goussevskaia et~al. \cite{GHWW09}, Halld{\'o}rsson and Wattenhofer \cite{HW09}, and Wan et al. \cite{WJY09} gave constant-factor approximation algorithms for the maximum-capacity problem yielding an $O(\log n)$-approximation algorithm for the scheduling problem, assuming uniform power. In \cite{GHWW09} they note that their $O(1)$-approximation algorithm also applies to the case where the ratio between the maximum and minimum power is bounded by a constant and for the case where the number of different power levels is constant.
More recently, Halld{\'o}rsson and Mitra \cite{HM11_SODA} have considered the case of oblivious power. This is a special case of non-uniform power where the power of a link is a simple function of the link's length. They gave an $O(1)$-approximation algorithm for the maximum capacity problem, yielding an $O(\log n)$-approximation algorithm for scheduling. 
Finally, the version where one assigns powers to the senders (i.e., with power control) has also been studied, see, e.g., \cite{AD09,MW06,K11,HM11_SODA,H12}.  For additional work on the subject, see \cite{AschnerCK14}.

\subsection{Our tools and goals}
\label{sec:tools}

Besides making progress on the actual problems being considered here, we view this work as another demonstration of what we hope to be a developing trend of combining tools from the computer algebra world with those of computational geometry and other fields.  Several relatively recent representatives of such synergy show examples of seemingly impossible speed-ups in geometric algorithms by expressing a subproblem in algebraic terms \cite{ma-cdplbf-16,ma-cdplbf-12,arst-ccvarp-07}.  The algebraic tools themselves are mostly classical ones, such as Fast Fourier Transform, fast polynomial multiplication, multipoint evaluation, and interpolation \cite{gg-mca-99,bp-pmcfa-94}; see Appendix~A for details. We combine them with only slightly newer tools from computational geometry, such as Voronoi diagrams, point-location structures in the plane, fast exact and approximate nearest-neighbor query data structures, and range searching data structures \cite{bcko-cgaa-08}; refer to Appendix~B.  One very recent result we need is that of Har-Peled and Kumar \cite{wann} that, as a special case, allows one to build a compact data structure for approximately answering multiplicatively weighted nearest-neighbor queries in the plane; the exact version appears to require building the classical multiplicatively weighted Voronoi diagram, which is a quadratic-size structure.

We hope that the current work will lead to further productive collaborations between computational geometry and computer algebra.

After preliminary versions of our work have appeared~\cite{ak-bpsdat-15,this-arXiv}, we became aware of the work of Ziegler \cite{Ziegler03} on approximate computation of forces and potentials in classical electrostatic and gravitational fields of multiple point sources. We found that the fundamental tools and techniques applied by Ziegler are quite similar to those applied by us in the first part of this manuscript (i.e., in Sections~\ref{sec:1d} and~\ref{sec:2d}); howoever, the context of SINR networks gives rise to several distinct challenges that distinguish between our work and his.

\subsection{Our results}
\label{sec:results}

We now summarize our main results.
We use $O^*$ notation to suppress logarithmic factors and $O_\eps$ to denote polynomial dependence on $1/\eps$, where $\eps>0$ is the approximation parameter.  We present algorithms for both the uniform-power and non-uniform-power settings, where the algorithms of the former type are usually somewhat simpler.
Throughout the paper we assume that $\alpha$ is an integer. Our algorithms can handle both even and odd values of $\alpha$, excluding the results in Sections~\ref{sec:trans-grid} and~\ref{sec:rec-grid} which require that $\alpha$ be even. Sometimes, the running time in the odd case is slightly slower than in the even case.
\begin{itemize}
\item In one dimension, we can perform $n$~queries among $n$ transmitters exactly in $O^*(n)$ total time; see Section~\ref{sec:1d}.
\item In two dimensions, we can perform $n$~queries among $n$ transmitters approximately in $O^*_\eps(n)$ total time; see Section~\ref{sec:approx}. 
\item If the underlying transmitters are using directional antennas, then we can perform $n$~queries among $n$ transmitters approximately in $O^*_\eps(n^{4/3})$ total time, after modifying the SINR inequality to suit this case; see Section~\ref{sec:dir_antennas}.
\item We can also facilitate exact batch queries when queries or transmitters form a grid; see Sections~\ref{sec:trans-grid} and~\ref{sec:rec-grid} for the exact statements.
\item
We apply our methods to speed up the preprocessing stage of the uniform-power point-location algorithm of Avin et al. \cite{aeklpr-sdciawn-12} mentioned above to $O^*_\eps(n)$; see Section~\ref{sec:speed-up}.
\item
The correctness of our faster construction is based on a couple of technical lemmas revealing interesting structural properties of SINR regions; see Section~\ref{sec:structural}.  
\end{itemize}

\section{Batched point location on the line}
\label{sec:1d}

In this section $\S$ is a set of $n\geq3$ point transmitters and $\Q$ is a set of $m$ query points, both on the line.
We first consider the \emph{uniform-power version} of the problem, where each transmitter has transmission power 1 (i.e., $p_1=\dots=p_n=1$), and then extend the approach to the arbitrary power version. 
\subsection{Uniform power}
\label{sec:1d-uni}

A query point $q$ receives $s_i$ if and only if
\[
  \frac{\frac{1}{|q-s_i|^\alpha}}{\Sigma_{j \ne i}{\frac{1}{|q-s_j|^\alpha}} + N} \geq \beta \, .
\]
Recall that, since $\beta > 1$, if $q$ receives one of the transmitters, then it must be the transmitter that is closest to it; we call it the \emph{candidate} transmitter for $q$ and denote it by $s(q)=s(q,\S)$.
Next, we define a univariate function $f$ as
\[
f(q) \coloneqq \sum_{j = 1}^n {\frac{1}{|q-s_j|^\alpha}}\, .
\]
Then, $q$ can hear its candidate transmitter $s(q)$ if and only if 
\[
E(q) \coloneqq \frac{\frac{1}{|q-s(q)|^\alpha}}{f(q) - \frac{1}{|q-s(q)|^\alpha} + N} \geq \beta \, .
\]
We distinguish between the case where $\alpha$ is even (Thereom~\ref{th:1d-batch-uniform}) and the case where it is odd (Theorem~\ref{th:1d-batch-uniform-odd}).
\begin{theorem}
  \label{th:1d-batch-uniform}
  For any fixed positive even integer $\alpha$, given a set $\S$ of
  transmitters (all of power~1) and a set $\Q$ of receivers, of sizes
  $n$ and $m$ respectively, we can determine which, if any,
  transmitter is received by each receiver in total time $O((n+m)
  \log^2 n \log \log n)$.
\end{theorem}

\begin{proof}
  As pointed out above, a receiver $q$ can receive only the closest
  transmitter $s(q)$, if any, as the SINR inequality implies
  $\frac{1}{{|q-s(q)|}^\alpha} > \frac{1}{{|q-s|}^\alpha}$ for any $s\neq s(q)$, or
  equivalently, $|q-s(q)|<|q-s|$.  So, as a
  first step, we identify the closest transmitter for each receiver,
  which can be done, for example, by sorting~$\S$, and using binary search for each
  receiver, in total time $O((m+n)\log n)$.  Moreover, we can compute
  the term $\frac{1}{{|q-s(q)|}^\alpha}$, for each $q \in Q$, in the same 
  amount of time.

  Observe that since $\alpha$ is even, we can view
  $f$ as a sum of $n$ low-degree fractional functions of a single real variable $q$ (i.e., the fractional functions $\frac{1}{(q-s_j)^\alpha}$, for $j=1,\ldots,n$), so 
  according to Corollary~\ref{cor:eval-sum-fractions}, we can now
  evaluate $f$ on all points of $\Q$ simultaneously in time $O((n+m)
  \log^2 n \log \log n)$.  

  In $O(m)$ additional operations we can evaluate the expressions
  $E(q_1),\dots,E(q_m)$ and determine for which receivers the SINR
  inequality holds, so that the signal is actually received.

  Computing and evaluating the fraction dominates the computation cost, so the
  total running time is $O((n + m) \log^2 n \log \log n)$.
\end{proof}

\begin{theorem}
  \label{th:1d-batch-uniform-odd}
  For any fixed positive odd integer $\alpha$, given a set $\S$ of
  transmitters (all of power~1) and a set $\Q$
  of receivers, of sizes $n$ and $m$ respectively, we can determine
  which, if any, transmitter is received by each receiver in total
  time $O((n + m) \log^3 n \log \log n)$.
\end{theorem}

\begin{proof}
  The proof is identical to that of Theorem~\ref{th:1d-batch-uniform} with one exception.
Again, we need to evaluate~$f$ for each receiver $q_i$, but now, unlike the previous case, we cannot view $f$ as a sum of proper rational functions, since, if we do so, the sign of some of the terms might be negative (rather than positive), depending on whether $q_i$ lies to the left or to the right of $s_j$. To overcome this difficulty, we add another layer of processing.

  We sort the transmitters $s_i$ by their position along the line and build a balanced binary search tree with transmitters as leaves, each node corresponding to a set of transmitters.  We associate with each node of the tree two sets of receivers, a \emph{right set} and a \emph{left set}, where each receiver in the right set (resp., left set) of a node $v$ lies to the right (resp., to the left) of all transmitters associated with $v$. We do this by searching for each receiver $q_i$ in the tree.  The set of all transmitters lying left of $q_i$ is naturally expressed as a disjoint union of $O(\log n)$ subsets corresponding to the nodes of the search tree, and we add $q_i$ to the right set of each of these nodes.
  Similarly, the set of all transmitters lying right of $q_i$ is naturally expressed as a disjoint union of $O(\log n)$ subsets corresponding to the nodes of the search tree, and we add $q_i$ to the left set of each of these nodes.
	
  Now, for each node of the tree, we apply Corollary~\ref{cor:eval-sum-fractions} twice; once for its right set of receivers and the partial sum of fractional functions of the form $\frac{1}{(q-s_j)^\alpha}$ associated with it, and once for its left set of receivers and the partial sum of fractional functions of the form $\frac{1}{(s_j-q)^\alpha}$ associated with it.
		
  We now proceed as in Theorem~\ref{th:1d-batch-uniform}, where $f(q_i)$ is obtained as a sum of a logarithmic number of terms that have already been computed. The total cost is $O(\log n)$ times that of Theorem~\ref{th:1d-batch-uniform}, as each transmitter and receiver appears in at most a logarithmic number of computations.
\end{proof}

Notice that the above ``trick'' is a simpler, one-dimensional version of the range-tree machinery that we use in Section~\ref{sec:approx}.

\subsection{Arbitrary power}
\label{sec:1d-non-uni}

We proceed in a similar manner, except for the construction of the multiplicatively weighted Voronoi diagram on a line, which is more subtle.

Observe that, if $q$ hears any of the transmitters in $\S$, then it is the \emph{candidate} transmitter $s(q)$, for which the expression $p(s)/|q-s|^\alpha$ is maximum, over all~$s \in \S$; $p(s)$ is the power of transmitter~$s$. 
Thus, we find for each $q \in \Q$ its corresponding transmitter, by
computing the appropriate weighted Voronoi diagram and querying it. This can be done in $O((n + m) \log n)$ time (see Fact~\ref{fact:vor-1d-wtd}).  
Next, we generalize the function $f$ defined above to mean
\[
f(q) \coloneqq \sum_{j = 1}^n {\frac{p_j}{|q-s_j|^\alpha}}\, .
\]
Now, $q$ hears $s(q)$ if and only if 
\[
E(q) \coloneqq \frac{\frac{p(s(q))}{|q-s(q)|^\alpha}}{f(q) - \frac{p(s(q))}{|q-s(q)|^\alpha} + N} \geq \beta \, ,
\]
and we have

\begin{theorem}
  \label{th:1d-batch-nonuniform}
  For any fixed positive integer $\alpha$, given a set $\S$ of
  transmitters (not necessarily all of the same power) and a set $\Q$
  of receivers, of sizes $n$ and $m$ respectively, we can determine
  which, if any, transmitter is received by each receiver in total
  time $O((n + m) \log^2 n \log \log n)$, if $\alpha$ is even, and $O((n + m) \log^3 n \log \log n)$, if $\alpha$ is odd.
\end{theorem}

\begin{proof}
  Assume $\alpha$ is even. For odd $\alpha$, the proof is identical, except for an additional layer of processing, which is described in the proof of Theorem~\ref{th:1d-batch-uniform-odd} and which adds a logarithmic factor to the total running time. 
  Instead of looking for the nearest transmitter to $q$, we need to
  look for the transmitter maximizing $p_i/|q-s_i|^\alpha$, which is
  the same as minimizing $|q-s_i|/p_i^{1/\alpha}$.  This corresponds to the
  one-dimensional \emph{multiplicatively weighted Voronoi diagram} of
  $\S$ with weights $p_i^{1/\alpha}$, which, as per Fact~\ref{fact:vor-1d-wtd} and surrounding discussion, has complexity $O(n)$ and can be constructed 
  in~$O(n \log n)$ time; see \cite{a-odwvd-86} and
  \cite[Theorem~6.1]{sa-dsstg-95}.  We once again use binary search to
  identify, for each receiver $q_i$, the only possible transmitter
  $s(q_i)$ that it may be able to hear, compute the term
  $p(s(q_i))/|q_i-s(q_i)|^\alpha$ needed to finish computing~$E(q_i)$, and thereby
  check the reception.

  The rest of the algorithm proceeds as in
  Theorem~\ref{th:1d-batch-uniform}: the function $f$ depends on the
  power of individual transmitters, but it is still a sum of $n$
  low-degree univariate fractions and the remainder of the analysis follows unchanged.
\end{proof}

\section{Batched point location in the plane}
\label{sec:2d}

In this section $\S=\{s_i\}$ is a set of $n$ point transmitters in the plane.
We consider three versions of batched point location, where in the
first two the answers we obtain are exactly correct, while in the
third one the answer to a query $q$ may be either \textsc{yes} (meaning
that $q$ receives $s$), \textsc{no} (meaning that $q$ does not receive any
transmitter), or \textsc{maybe} (meaning that $q$ may or may not be receiving some 
transmitter; the SIN~ratio is too close to~$\beta$ and we are unable to decide quickly whether it is above or below $\beta$).   

Specifically, we consider the following three versions of batched point location.
In the first version, we assume that the \emph{transmitters} form an $\sqrt{n} \times \sqrt{n}$ non-uniform grid and
that each transmitter has power~1. We show how to solve a \emph{single} point-location query in this setting
in $O(\sqrt{n} \log^2 n \log \log n)$ (rather than linear) time; refer to Section~\ref{sec:trans-grid}.
In the second version (Section~\ref{sec:rec-grid}),  
we assume that the \emph{receivers} form an $n \times n$ non-uniform grid, but
the $n$ transmitters, on the other hand, are located anywhere in the plane.
Moreover, we allow arbitrary transmission powers.
We show how to answer the $n^2$ queries in near-quadratic (rather than cubic) time.

Finally, in the third version (Section~\ref{sec:approx}), we do not make any assumptions on the location of the devices (either transmitters or receivers). As a result of this, we might not be able to give a definite answer in borderline instances. Specifically, given $n$ transmitters and $m$ receivers, we compute (in total time near-linear in $n+m$), for each receiver $q$, its unique candidate transmitter $s$ and a value $\EE(q)$, such that (a)~if $\EE(q)$ is sufficiently greater than $\beta$, then $q$ surely receives $s$, (b)~if $\EE(q)$ is sufficiently smaller than $\beta$, then $q$ surely does not receive $s$, and (c)~otherwise, $\EE(q)$ lies in the \emph{uncertainty interval}, and $q$ may or may not receive some transmitter . 
We first present a solution for which the uncertainty interval is $[2^{-\alpha/2}\beta,2^{\alpha/2}\beta)$, i.e., a constant-factor approximation. We then generalize it so that the uncertainty region is $[(1-\eps)\beta,(1+\eps)\beta)$, for any $\eps > 0$, i.e., a PTAS.
We consider both the uniform- and arbitrary-power settings.

\subsection{General discussion}
Once again, the SINR inequality determines which, if any, of the transmitters 
$s \in S$ can be heard by a receiver at point $q$ and the only
candidate transmitter $s(q)$ is the one that minimizes
$|q-s|/p^{1/\alpha}$ among all transmitters $s$ with corresponding
power $p$.  In the uniform-power case, this means the transmitter
closest to $q$ in Euclidean distance, and the matching space
decomposition is the Euclidean Voronoi diagram which can be
constructed in $O(n \log n)$ time (see Section~\ref{sec:voronoi}),
where $n=|\S|$.  In the non-uniform-power case, this corresponds to
the multiplicatively weighted Voronoi diagram in the plane, which is a
structure of worst-case complexity $\Theta(n^2)$ that can be
constructed in time $O(n^2)$; see \cite{ae-oacwvdp-84}.

Once again we define the function $f(q)$, which represents the total signal strength
at $q$ from \emph{all} transmitters, and express the decision of whether
the transmitter $s(q)$ is received at $q$ by computing~$E(q)$ from $f(q)$ and $s(q)$ and comparing
it with $\beta$.  The difference from the one-dimensional case is that
$f(q)$ is now a sum of low-degree \emph{bivariate} fractions, with
the two variables being the coordinates of $q$. 

In all cases, the goal is to evaluate $f(q)$, for each receiver $q$, and
to identify the suitable candidate transmitter $s(q)$,
faster than by brute force.  Given this information, the
decision can be made in constant time per receiver.

We now present each two-dimensional problem in turn.

\subsection{Transmitters on a grid}
\label{sec:trans-grid}

In this version we assume that the transmitters of $\S$ form an $\sqrt{n} \times \sqrt{n}$
non-uniform grid and have uniform power.  We show how to answer a \emph{single}
arbitrary point-location query in roughly $\sqrt{n}$, rather than linear,
time.
We assume that $\S=X \times Y$, where $X$ and $Y$ are two sets of $\sqrt{n}$
numbers (coordinates) each.  We start by sorting $X$ and $Y$.  

As mentioned above, the problem reduces to computing $f(q)$ and
identifying~$s(q)$ for the query~$q$.  The latter is not
difficult: since all powers are the same, it is
sufficient to identify the point in $X \times Y$ closest to $q$.  It is
easy to check that the closest point is always one of the (at most) four
corners of the grid cell containing $q$ and can be found by one binary search on
$X$ and one on $Y$ and comparing each of the at most four corresponding
distances to $q$. 

So we focus on the computation of $f(q)$.  We rewrite it as
\[
  f(q;X;Y) = \sum_{x\in X} \sum_{y\in Y} {\frac{1}{|q-(x,y)|^\alpha}} = 
  \sum_{x\in X} g_Y (q;x),
\] 
where we view
\[
  g_Y(q;x) \coloneqq \sum_{y\in Y} {\frac{1}{|q-(x,y)|^\alpha}}
\]
as a sum of fractional functions of $x$ only, having substituted the actual
values for the numbers of $Y$.  Using
Corollary~\ref{cor:eval-sum-fractions}, we can sum these fractions symbolically and then 
evaluate the sum at the $\sqrt{n}$ distinct values $x_1, x_2, \dots, x_{\sqrt{n}}$ of $X$,
obtaining $g_Y(q;x_1), g_Y(q;x_2), \dots, g_Y(q;x_{\sqrt{n}})$ and thereby $f(q)$, in 
$O(\sqrt{n} \log^2 n \log \log n)$ operations, which dominates the running
time.

\begin{theorem}
  \label{th:2d-transmitters-on-grid}
  For any fixed positive even integer $\alpha$, given a set $\S$ of $n$
  transmitters (all of power~1) forming a $\sqrt{n} \times \sqrt{n}$ non-uniform grid and
	a receiver $q$, we can determine
  which, if any, transmitter, is received by $q$ in
  time $O(\sqrt{n} \log^2 n \log \log n)$.
\end{theorem}

\begin{notes}
  (a)~For a $k\times l$ grid, with $l \leq k$, this would take $O((k+l) \log^2 l \log \log
  l) + O((k+l)\log (k+l)) = O(k\log^2 l \log \log l + k \log k)$ time.

  (b)~Odd integer $\alpha$ cannot be handled using these methods, as $|q-(x,y)|^\alpha$ is not a polynomial (informally, it involves square roots).
\end{notes}

\subsection{Receivers on a grid}
\label{sec:rec-grid}
In this version we assume that the receivers of $\Q$ form an $n \times n$ non-uniform grid.
The $n$ transmitters, on the other hand, are located anywhere in the plane. Moreover, we allow arbitrary transmission powers.
We show how to answer the $n^2$ queries in near-quadratic (rather than  cubic) time.

In this case, we need to evaluate $f(q)$, which is a sum of $n$
bivariate low-degree fractions at all points $q_{ij}$ of a
two-dimensional, possibly non-uniform grid $X \times
Y$.\footnote{Slanted and sheared grids can be handled by a simple
  extension of this observation; we omit the easy details.}
This is taken care of in~$O(n^2 \log^2 n \log \log n)$ time by
Corollary~\ref{cor:eval-sum-fractions-bi}.

The only missing ingredient is identifying the candidate transmitter
$s(q)$ for each $q \in \Q$.  This can be done, for example,
by computing the weighted Voronoi diagram, preprocessing it for point
location, and executing $n^2$ such queries, one for each $q \in \Q$.  

It turns out that the following alternative may be 
simpler to implement: Observe that on each line $x=x_j$, the
functions $|q-s_i|/p_i^{1/\alpha}$ behave very similarly to the
univariate case in Section~\ref{sec:1d-non-uni} above, so the problem
can be solved in $O(n \log n)$ time per line, for a total of $O(n^2
\log n)$ time; refer to Fact~\ref{fact:vor-2d-slice}.

Combining the computation of $f(q)$ over all points of $X \times Y$
and identification of $s(q)$, we obtain the claimed result in
$O(n^2 \log^2 n \log \log n)$ time.

\begin{theorem}
  \label{th:2d-receivers-on-grid}
  For any fixed positive even integer $\alpha$, given a set $\S$ of $n$
  transmitters (not necessarily all of the same power), and a set $\Q$ of $n^2$ receivers forming an $n \times n$ non-uniform grid, we can determine
  which, if any, transmitter, is received by each receiver in total time $O(n^2 \log^2 n \log \log n)$.
\end{theorem}

\begin{notes}
  (a)~There are two obstacles to speeding this up for a smaller number of receivers: (i)~An explicit representation of $f(q)$ has $\Theta(n^2)$ coefficients and thus cannot be processed in subquadratic time.  (ii)~It seems harder to identify the candidate transmitter in the non-uniform power case; the first solution builds a quadratic-size space decomposition and the second requires linear time for each set of collinear points; both result in near-quadratic performance.  A new idea is needed.

  (b)~Odd integer $\alpha$ cannot be handled using these methods, as $|q-(x,y)|^\alpha$ is not a polynomial (informally, it involves square roots of polynomials).
\end{notes}

\begin{note}
  If we do not require the $n^2$ receivers to lie on an $n\times n$ grid, we can still obtain a subcubic solution by applying a different algebraic technique. Indeed, according to Fact~\ref{fact:bi:general-eval}, we can determine which, if any, transmitter, is received by each receiver in total time $O(n^{1+\omega_2/2+\eps})$, for any $\eps>0$, where $\omega_2<3.334$ is a constant related to the efficiency of matrix multiplication.    
\end{note}

\subsection{Approximating the general case}
\label{sec:approx}

We now abandon the ambition to get exact answers and aim
for an approximation algorithm, in the sense we will make precise
below.  Again, $\S=\{s_i\}$ is the set of $n$ transmitters, with each
$s_i$ a point in the plane with power $p_i$; similarly $\Q=\{q_j\}$ is the
set of $m$ receivers (queries), where a generic receiver is $q=(q_x,q_y)$.

For a query point $q$ and a transmitter $s=(s_x,s_y)$ of power $p$, set $l(q,s) = \max\{|q_x
- s_x|, |q_y - s_y|\}$; in other words, $l(q,s)$ is the $L^\infty$
distance between points $q$ and $s$.  In complete analogy to our
previous approach,
put
\[
\ff(q) \coloneqq \sum_{i=1}^n\frac{p_i}{l(q,s_i)^\alpha} \ \ \ \ \mbox{and} \ \ \ \
\EE(q) \coloneqq \frac{\frac{p}{l(q,s)^\alpha}}{\ff(q) - \frac{p}{l(q,s)^\alpha} + N} \, .
\]

What is the significance of the quantity $\EE(q)$?  Since for any two
points $s,q$, $l(q,s) \leq |q-s| \leq \sqrt{2}\, l(q,s)$, 
\[
2^{-\alpha/2} \frac{p_j}{l(q,s_j)^\alpha} \leq 
\frac{p_j}{|q-s_j|^\alpha} \leq 
\frac{p_j}{l(q,s_j)^\alpha},
\]
so $2^{-\alpha/2}\ff(q)\leq f(q) \leq \ff(q)$, and therefore
$2^{-\alpha/2}\EE(q)\leq E(q) \leq 2^{\alpha/2}\EE(q)$.
Informally, $\EE(q)$ is ``pretty close'' to $E(q)$.  

This suggests an approximation strategy that begins by computing
$\EE(q)$ instead of $E(q)$.  If $\EE(q)\geq 2^{\alpha/2} \beta$, we know that
$E(q)\geq\beta$ and the signal from the unique
candidate transmitter $s(q)$ \emph{is} received (\textsc{yes}).  If $\EE(q)<2^{-\alpha/2}
\beta$, then $E(q)<\beta$ and the signal from $s(q)$ is \emph{not} received
and therefore no signal is received by $q$ (\textsc{no}).  For
intermediate values of $\EE(q)$, we cannot definitely determine whether
$s(q)$'s signal is received at $q$ (\textsc{maybe}). 

Now we turn to the actual batch computation of $\EE(q)$ for all
receivers in $\Q$ and point out a few additional caveats.

Computationally, $\EE(q)$ can be evaluated in constant time, given
$\ff(q)$ and point $s(q)=s(q,\S)$.
So we focus on these two subproblems.
For the uniform-power case, we can construct the Voronoi diagram of $\S$,
preprocess it for point location, and query it for $s(q)$ with each receiver~$q$,
for a total cost of $O((n+m)\log n)$; see Fact~\ref{fact:vor-2d}
and \cite{bcko-cgaa-08}.
In the case of non-uniform power, if we are content with near-quadratic
running time, we can determine $s(q)$ by computing the
multiplicatively weighted Voronoi diagram of $\S$ as outlined above, and
then querying it with each receiver in total time
$O(n^2+m \log n)$ (see Fact~\ref{fact:vor-2d-wtd} and  \cite{bcko-cgaa-08,ae-oacwvdp-84}), which is too expensive when $m \approx n$.  We provide an alternative below.

We show how to compute the values $\ff(q_1), \dots, \ff(q_m)$ in
near-linear time, using a two-dimensional orthogonal range search
tree.  Indeed, observe that $l(s,q)=|q_x - s_x|$ provided $|q_x - s_x|\geq
|q_y - s_y|$.  For a fixed $q$, the region $W_q$ containing the transmitters
of $\S$ satisfying this inequality is a $90^\circ$ double wedge.
Using (a $45^\circ$~tilted version of) the orthogonal range search tree \cite{bcko-cgaa-08}
(see Section~\ref{sec:ortho}), we can construct a pair decomposition
$\{(\S_i,\Q_i)\}$ of small size, so that each pair $(s,q)$ with $s \in
W_q$ appears in exactly one product $\S_i \times \Q_i$.  

We now denote by $\ff(q,Z)$ the sum analogous to $\ff(q)$, where
the summation goes over the elements of the supplied set $Z$ rather than those
of $\S$.  Clearly,
\begin{equation}
  \label{eq:partial-sums}
  \ff(q,\S \cap W_q)=\sum_{i:q \in \Q_i} \ff(q,\S_i),
\end{equation}
by the definition of the pair decomposition.  The number of terms in the
last sum is $O(\log^2 n)$.  Notice that $\ff(q,\S_i)$, for a fixed $i$,
is a sum of small fractional \emph{univariate} functions, with $|\S_i|$
terms in it, since the expression for transmitters in $W_q$ depends only on
$q_x$ and not on $q_y$.  Now for each pair $(\Q_i,\S_i)$, we use
Corollary~\ref{cor:eval-sum-fractions} to batch evaluate $\ff(q,\S_i)$ on each
$q\in \Q_i$ in total time $O((|\Q_i| + |\S_i|) \log^2 |\S_i| \log \log
|\S_i|) = O((|\Q_i| + |\S_i|) \log^2 n \log \log n)$.  This gives us all
the summands of~\eqref{eq:partial-sums} and therefore allows us to
evaluate $\ff(q, \S \cap W_q)$ for all $q \in \Q$, in total time at most
proportional to $\sum_i (|\Q_i| + |\S_i|) \log^2 n \log \log n = (\sum_i
(|\Q_i| + |\S_i|)) \log^2 n \log \log n = O((m+n)\log^4 n \log \log n)$. 

Of course, we have only treated those $s$ that lie in $W_q$.  But the calculation
is repeated in the complementary double wedge,
where now only the $y$-coordinates matter and $\ff(q)$ is the sum of the two values thus obtained.

\begin{theorem}
  \label{th:2d-general-uniform}
  For any fixed positive even integer $\alpha$, given a set $\S$ of $n$
  transmitters (all of power~1) and a set $\Q$ of $m$ receivers, we can do the following in 
  total time $O((m+n) \log^4 n \log \log n)$ and $O((m+n) \log^2 n)$ space.
  For each $q \in \Q$, we find 
  its unique candidate transmitter $s(q)$ and compute a value $\EE(q)$, such that (i)~if $\EE(q) \ge 2^{\alpha/2}\beta$,
  then $q$ can definitely hear $s(q)$, (ii)~if $\EE(q) < 2^{-\alpha/2}\beta$, then $q$ definitely cannot hear $s(q)$, and 
  (iii)~if $2^{-\alpha/2}\beta \le \EE(q) < 2^{\alpha/2}\beta$, then $q$ may or may not hear $s(q)$.
\end{theorem}

The algorithm for the non-uniform power case is hampered by the fact
that the obvious way to identify the candidate transmitter each
receiver might hear seems to involve constructing the multiplicatively
weighted Voronoi diagram of quadratic complexity.  However, we do not
need the exact multiplicatively closest neighbor, but rather a
reasonably-close approximation of the value
$|q-s|/p(s)^{1/\alpha}$, over all $s \in \S$ (being off by a
multiplicative factor of at most $2^{1/2}$ is sufficient; see the discussion below).
Such an approximation is provided by the first algorithm in Fact~\ref{fact:2d-ann-wtd} (see Har-Peled and Kumar \cite{wann-focs,wann}), for a constant value of the approximation parameter $\eps$ (namely, $\eps = 2^{1/2}-1$), yielding the following:

\begin{theorem}
  \label{th:2d-general}
  For any fixed positive even integer $\alpha$ and any $\beta>2^{\alpha/2}$, given a set $\S$ of $n$
  transmitters of arbitrary powers and a set $\Q$ of $m$ receivers, we can do the following in 
  total time $O(n \log^7 n + m \log^4 n \log \log n)$ and $O(n \log^4 n + m \log^2 n)$ space: For each $q \in \Q$, we find 
  a transmitter $s_q$ and compute a value $\EE(q)$, such that (i)~if $\EE(q) \ge 2^{\alpha/2}\beta$,
  then $q$ can definitely hear $s_q$ (implying that $s_q=s(q)$), (ii)~if $\EE(q) < 2^{-\alpha/2}\beta$, then $q$ definitely cannot hear any transmitter, and 
  (iii)~if $2^{-\alpha/2}\beta \le \EE(q) < 2^{\alpha/2}\beta$, then $q$ may or may not hear one of the transmitters.
\end{theorem}

\begin{note}
  The transmitter $s_q$ in the theorem above is not necessarily the unique candidate transmitter $s(q)$. We would like to show that if $\EE(q) \ge 2^{\alpha/2}\beta$ (and therefore $E(q) \ge \beta$), then $s_q$ is necessarily $s(q)$. Assume that they are different (i.e., that $s_q \ne s(q)$), and let $e_q$ (resp., $e(q)$) be the strength of $s_q$'s signal (resp., $s(q)$'s signal) at $q$. Then, we know that $e_q \leq e(q) \leq 2^{\alpha/2} e_q$. Notice that $E(q) \le e(q)/e_q$, since $E(q)$ is maximized when there is no third transmitter and no noise, so $e(q)/e_q \ge \beta$ since $E(q) \ge \beta$. Recall that we are assuming that $\beta > 2^{\alpha/2}$, so we get that $e(q)/e_q > 2^{\alpha/2}$, which is a contradiction. 

  Therefore, if fact, it must be the case that $s_q=s(q)$, as claimed.
\end{note}

We now turn the algorithm described above into a PTAS, in the sense that we will confine $\EE(q)$ to the range $((1-\eps)E(q),(1+\eps)E(q)]$,
for a prespecified $\eps>0$.  We outline the approach below.
Consider the regular $k$-gon $K_k$ circumscribed around the Euclidean unit disk, for a large enough even $k\geq 4$ specified below.  We modify the above algorithm, replacing the $L^\infty$-norm whose ``unit disk'' is a $2\times2$~square, with the norm $|\dots|_k$ with $K_k$ as the unit disk. Then 
$|v|_k \leq |v| \leq (1+\Theta(k^{-2}))|v|_k$, for any vector $v$ in the plane.  In the range-searching data structure, wedges with opening angle $\pi/2=2\pi/4$ are replaced by wedges with opening angle $2\pi/k$, and we need $k/2$ copies of the structure.

In terms of the quality of approximation, the factor 
$2^{\alpha/2} = (\sqrt2)^\alpha$ is replaced by
$(1+\Theta(k^{-2}))^\alpha \approx 1+\alpha \Theta(k^{-2})$. 
Hence to obtain an approximation factor of $1+\eps$, we set
$\eps = \alpha \Theta(k^{-2})$, or $k=c(\alpha/\eps)^{1/2}$, for a suitable absolute constant $c$.
In other words, it is sufficient to create $O(\eps^{-1/2})$ copies of the data structure.  To summarize, we have:

\begin{theorem}
  \label{th:2d-general-uniform-ptas}
  For a positive $\eps$, any fixed positive even integer $\alpha$, given a set $\S$ of $n$
  transmitters (all of power~1) and a set $\Q$ of $m$ receivers, we can do the following in 
  total time $O((m+n) \eps^{-1/2} \log^4 n \log \log n)$ and $O((m+n) \eps^{-1/2} \log^2 n)$ space.
  For each $q \in \Q$, we find 
  its unique candidate transmitter $s(q)$ and compute a value $\EE(q)$, such that (i)~if $\EE(q) \ge (1+\eps)\beta$,
  then $q$ can definitely hear $s(q)$, (ii)~if $\EE(q) < (1-\eps)\beta$, then $q$ definitely cannot hear $s(q)$, and 
  (iii)~if $(1-\eps)\beta \le \EE(q) < (1+\eps)\beta$, then $q$ may or may not hear $s(q)$.
\end{theorem}

\begin{theorem}
  \label{th:2d-general-ptas}
  For a positive $\eps<1 -\beta$, any fixed positive even integer $\alpha$,\footnote{%
    This requirement is analogous to that in Theorem~\ref{th:2d-general} to guarantee that the approximately highest-strength transmitter returned by the data structure is in fact the right one.}
  given a set $\S$ of $n$
  transmitters of arbitrary powers and a set $\Q$ of $m$ receivers, we can do the following in 
  total time $O(n \eps^{-6} \log^7 n + m \eps^{-1/2}\log^4 n \log \log n)$ and $O(n \eps^{-6}\log^4 n + m \eps ^{-1/2}\log^2 n)$ space: For each $q \in \Q$, we find 
  a transmitter $s_q$ and compute a value $\EE(q)$, such that (i)~if $\EE(q) \ge (1+\eps)\beta$,
  then $q$ can definitely hear $s_q$ (implying that $s_q=s(q)$), (ii)~if $\EE(q) < (1-\eps)\beta$, then $q$ definitely cannot hear any transmitter, and 
  (iii)~if $(1-\eps)\beta \le \EE(q) < (1+\eps)\beta$, then $q$ may or may not hear one of the transmitters.
\end{theorem}

\begin{notes}
  (a) Hereafter, we refer to cases (i), (ii), and (iii) as our approximate SINR diagram point location oracle returning \textsc{yes}, \textsc{no}, and \textsc{maybe}, respectively.

  (b)~Most likely the data structure from \cite{wann} that we use for approximate multiplicatively weighted nearest neighbors is not the best possible for our purposes, as it is designed to handle a more general situation.

  (c)~Also, we could have replaced the exact Euclidean Voronoi diagram in Theorems~\ref{th:2d-general-uniform} and~\ref{th:2d-general-uniform-ptas} by an approximate nearest-neighbor structure.  While unnecessary in the plane, it would allow for polylogarithmic amortized batched queries in 3D. In fact, everything we described in this section generalizes to 3D and, moreover, to $\Reals^d$, for any constant $d>1$, to yield amortized polylogarithmic-time batched queries.  One needs to approximate the Euclidean ball by a suitable polyhedron, use $d$-dimensional range trees and the $d$-dimensional data structure from \cite{wann} for multiplicatively-weighted approximate neighbors (in the uniform case, one can use any one of the simpler approximate nearest-neighbor structures, such as the Approximate Voronoi Diagram; see further discussion in \cite{wann}).

  (d)~The algorithms in Theorems~\ref{th:2d-general-uniform}, \ref{th:2d-general}, \ref{th:2d-general-uniform-ptas}, and \ref{th:2d-general-ptas} extend to odd values of~$\alpha$ with very minor modifications.  The difficulty with such values is that an expression such as $|q_x-s_{j,x}|^\alpha$ is not a polynomial, as it is equal to $\pm(q_x-s_{j,x})^\alpha$ depending on whether $s_j$ is located to the left or to the right of $q$, in the double wedge $W_q$: we need to know if $s_j$ lies in the left wedge of $W_q$ or in the right one.  In other words, the algorithm will proceed as before (with an appropriate sign correction), as long as we replace double wedges by single wedges in the range searching data structure.  (In fact, the range search already uses single wedges internally, so very little actual change is required.)
\end{notes}

\subsection{Handling directional antennas}
\label{sec:dir_antennas}

Heretofore
we have assumed that the transmitters in $\S$ use omni-directional antennas.  If directional antennas are employed instead, the SINR inequality must be modified accordingly.
The coverage area of the directional antenna of transmitter $s$ is a wedge $W=W(s)$ of angle $\theta=\theta(s)$ and apex at~$s$. 
Only points within $W$ can receive $s$.
Assuming that $q \in W$, 
we need to evaluate the left side of the SINR inequality, taking into account only the transmitters of $\S$ whose wedges cover $q$.
That is, $q$~receives $s$ of power $p$ if and only if $q \in W(s)$ and 
\[
\frac{\frac{p}{|q-s|^\alpha}}{\Sigma_{s_j \ne s:q \in W_j}{\frac{p_j}{|q-s_j|^\alpha}} + N} \geq \beta \, .  
\]

The results of Section~\ref{sec:approx} are still relevant (as we show below), but they cannot be applied immediately.  We first need to construct a data structure for range searching with wedges among the set of $n$ points $\Q$:  Given a query wedge $W=W(s)$, return $\Q \cap W$ as a union of a small number of canonical subsets of $\Q$. Then, we perform a query with each of the wedges $W(s)$, $s \in \S$, to obtain a collection of pairs of sets $\{(\S_i,\Q_i)\}$ with the property that,
for any $(s,q) \in \S \times \Q$, if $q \in W(s)$, then there exists a single pair $(\S_i, \Q_i)$ such that $(s,q) \in \S_i \times \Q_i$, and, if $q \notin W(s)$, then $(s,q) \notin \S_i \times \Q_i$, for any $i$. We now can apply the results of Section~\ref{sec:approx} to each of the pairs $(\S_i,\Q_i)$ to obtain theorems analogous to Theorems~\ref{th:2d-general-uniform}--\ref{th:2d-general-ptas}. The time bound in each of these theorems will be $O^*_\eps(n^{4/3})$, assuming $|\S|=|\Q|=n$. We omit the remaining details.  

Note that the additional overhead of dealing with directional antennas essentially disappears if there is only a fixed set of directions bounding the wedges $W(s)$.

\section{Some structural properties of SINR regions}
\label{sec:structural}

In this section we obtain results on the sensitivity of the reception regions to slight changes in the reception threshold. Specifically, we prove two technical lemmas (assuming uniform power), which roughly state that scaling a reception region $R_i$ centered at $s_i$ by a factor of $(1+\eps)$ (respectively, $(1-\eps)$) is equivalent to decreasing $\beta$ (respectively, increasing $\beta$) by a factor of $(1-\Theta(\eps))$ (respectively, $(1+\Theta(\eps))$. We use these lemmas in the next section to speed up the construction time of a previous data structure for approximate point location in the SINR diagram. These lemmas also imply that the two notions of approximating $R_i$ mentioned in the introduction are actually equivalent, at least for the uniform power case; see the discussion below.

We introduce the following auxiliary notation: First, for a positive number $\delta$, we let $\delta R_i$ be a copy of the set $R_i$ scaled by a factor $\delta$; the center of the scaling transformation is $s_i$, the location of the $i$th site.  For example, $2 R_i$ is a copy of $R_i$ scaled so that its linear dimensions are twice as large, with the center at $s_i$. Notice that $\area(\delta R_i) = \delta^2 \area(R_i)$. 

Consider modifying an instance of the SINR problem, in the following manner: Keep the locations and the power of the transmitter~$s_i$ the same, but replace the ratio threshold $\beta$ by some value %
$\gamma>1$.  We use $R_i(\gamma)$ to denote the region of $s_i$ in the resulting SINR diagram.  By definition, $R_i=R_i(\beta)$.

\begin{lemma}
  \label{lem:tedious}
  Consider a set of $n$ uniform-power transmitters at locations $s_i$, $i=1,\dots,n$. 
  Assume $\beta>1$. There exist values $c_2 > c_1 > 0$, that depend only on $\alpha$ and $\beta$, such that, for any $\eps \in (0,1/c_2)$ and all $i$,\footnote{%
    The first inclusion holds for any $\eps \in (0,1/c_1)$.}
  \[
    R_i((1-c_1\eps)\beta) \subset (1+\eps)R_i \subset R_i((1-c_2\eps)\beta).
  \]
\end{lemma}
\begin{proof}
  Fix a transmitter $s_i$.  Fix a point $q$ on the boundary of $R_i$.  Its SIN~ratio $E(q)$ relative to $s_i$, defined as
  \[
    E(q)\coloneqq \frac{\frac{1}{|q-s_i|^\alpha}}{\Sigma_{j \ne
        i}{\frac{1}{|q-s_j|^\alpha}} + N},
  \]
  is precisely $\beta$, since $R_i$ is convex
  \cite{aeklpr-sdciawn-12}.  Without loss of generality, $s_i$ lies at
  the origin and $q$ on the positive $x$-axis at distance $d$ from it.
  Put $d'=(1+\eps)d$.  Consider the point $q'=(d',0)$ on the boundary
  of $(1+\eps)R_i$.  We measure how different the SIN~ratio $E(q')$ with respect to $s_i$ is
  from $E(q)$ by bounding $E(q')/E(q)$.  Let $N(q)$ and $D(q)$ denote the numerator and denominator of the above expression for $E(q)$, respectively.
  Then, $E(q')/E(q) = (N(q')/N(q)) \cdot (D(q)/D(q'))$, and
  $N(q')/N(q) = (1/d')^\alpha/(1/d)^\alpha=(d/d')^\alpha = (1+\eps)^{-\alpha}$. 
  We now consider the ratio of the denominators (i.e., $D(q)/D(q')$).

  We first examine the effect of the noise $N$ on the ratio of the denominators.  It is easy to check that the ratio is a monotone function of~$N$ and approaches~$1$ (either from below or from above) as $N$ grows arbitrarily large.  Hence the extreme values of the ratio of the denominators can be computed by setting $N=0$; one of the extremes corresponds to the denominator decreasing as we move from $q$ to $q'$ and the other to it increasing.

  The maximum possible \emph{decrease} in the denominator
  corresponds to the maximum \emph{increase} in the distances to all
  sites other than $s_i$, which in turn corresponds to every such
  site~$s_j$ lying on the negative $x$-axis (we will see below that $s_j$
  cannot lie on the segment $s_iq$).  Thus,
  \[
     \frac{D(q)}{D(q')} = \frac{\Sigma_{j \ne i}{\frac{1}{|q-s_j|^\alpha}}}{\Sigma_{j \ne i}{\frac{1}{|q'-s_j|^\alpha}}} \le
     \max_{j \ne i}\frac{\frac{1}{|q-s_j|^\alpha}}{\frac{1}{|q'-s_j|^\alpha}} = \max_{j \ne i}(\frac{|s_j-q'|}{|s_j-q|})^\alpha =
     \max_{j \ne i}(\frac{|s_j-q|+|q'-q|}{|s_j-q|})^\alpha \ .
  \]
  We now observe that, since $E(q)=\beta>1$, not only
  $|s_i-q| \leq |s_j-q|$ for all $j\neq i$, but also
  $|s_i-q| \beta^{1/\alpha} \leq |s_j-q|$ (this follows from the fact that, since the SIN~ratio of $q$ is $\beta$, the signal strength of $s_i$ at $q$ is at least $\beta$ times that of $s_j$ at $q$, i.e., $1/|s_i-q|^\alpha \geq \beta/|s_j-q|^\alpha$).\footnote{
    Notice that this implies $s_j \not\in s_iq$, as promised.}
  Therefore, for all $j\neq i$,
  \[
    \frac{|s_j-q|+|q'-q|}{|s_j-q|} \leq \frac{\beta^{1/\alpha} |s_i-q|+\eps|s_i-q|}{\beta^{1/\alpha}
      |s_i-q|} =
    {1+\eps\beta^{-1/\alpha}},
  \]
  and
  \begin{align*}
    \frac{E(q')}{E(q)} & = (1+\eps)^{-\alpha} \cdot \frac{D(q)}{D(q')} \le
    (1+\eps)^{-\alpha} \cdot (1+\eps\beta^{-1/\alpha})^\alpha \\
    & = (\frac{1+\eps\beta^{-1/\alpha}}{1+\eps})^\alpha 
    = (\frac{1+\eps - (\eps-\beta^{-1/\alpha}\eps)} {1+\eps})^\alpha\\
    &
      = (1 -  \frac { (1  - \beta^{-1/\alpha})\eps } { 1 + \eps})^\alpha
      \le  (1 -  \frac { 1  - \beta^{-1/\alpha } }{ 2}  \eps)^\alpha
      \leq 1-\frac{(1-\beta^{-1/\alpha})}{2}\eps,
  \end{align*}
  provided $\eps \leq 1$ and $\alpha\geq 1$.
   
  At the other extreme, the distance to $s_j$ is maximally decreased if it lies on the positive $x$-axis beyond $q'$.  So we assume that this is the case for all sites $s_j$, $j\neq i$. Thus,
\[
     \frac{D(q)}{D(q')} \ge
     \min_{j \ne i}\frac{\frac{1}{|q-s_j|^\alpha}}{\frac{1}{|q'-s_j|^\alpha}} = \min_{j \ne i}(\frac{|s_j-q'|}{|s_j-q|})^\alpha =
     \min_{j \ne i}(\frac{|s_j-q|-|q'-q|}{|s_j-q|})^\alpha \ .
  \]  
Once again using $|s_i-q| \beta^{1/\alpha} \leq |s_j-q|$, we obtain, for all $j\neq i$,
  \[
    \frac{|s_j-q|-|q'-q|}{|s_j-q|} \geq \frac{\beta^{1/\alpha} |s_i-q|-\eps|s_i-q|}{\beta^{1/\alpha}|s_i-q|} =
    {1-\eps\beta^{-1/\alpha}},
  \]
  yielding
  \begin{align*}
    \frac{E(q')}{E(q)} & = (1+\eps)^{-\alpha} \cdot \frac{D(q)}{D(q')} \ge
    (1+\eps)^{-\alpha} \cdot (1-\eps\beta^{-1/\alpha})^\alpha  
    = (\frac{1-\eps\beta^{-1/\alpha}}{1+\eps})^\alpha \\ 
    & = (\frac{1+\eps-(\eps+\eps\beta^{-1/\alpha})}{1+\eps})^{\alpha}
    = (\frac{1+\eps-(1+\beta^{-1/\alpha})\eps}{1+\eps})^{\alpha} \\
    & = (1 - \frac{1+\beta^{-1/\alpha}}{1+\eps}\eps)^{\alpha} 
    \geq (1 - (1+\beta^{-1/\alpha})\eps)^{\alpha}
    \geq 1-\alpha(1+\beta^{-1/\alpha})\eps,
  \end{align*}
  where we used $\eps>0$ in the second to last step and $(1-x)^c\geq 1-xc$, for $0<x<1$ and $c\geq1$, in the last step.  The last estimate is only helpful if $\alpha(1+\beta^{-1/\alpha})\eps<1$, or $\eps <1/(\alpha(1+\beta^{-1/\alpha}))$.

  Now recall $R_i$ and $(1+\eps)R_i$ are star-shaped with center $s_i$, and $q$ (respectively, $q'$) is a boundary point of $R_i$ (respectively, $(1+\eps)R_i$).  We have therefore proven that 
  \[
    R_i((1-c_1\eps)\beta) \subset (1+\eps)R_i \subset R_i((1-c_2\eps)\beta),
  \]
  with $c_1=\frac{1-\beta^{-1/\alpha}}{2}$ and $c_2=\alpha(1+\beta^{-1/\alpha})$, as promised.
\end{proof}

We will also need an analogous statement for scaling $R_i$ down rather than up.
\begin{lemma}
  \label{lem:tedious2}
  Consider a set of $n$ uniform-power transmitters at locations $s_i$, $i=1,\dots,n$. 
  Assume $\beta>1$. There exist values $c_4 > c_3 > 0$ that depend only on $\alpha$ and $\beta$ so that, for any $\eps \in (0,1/2)$, 
  \[
    R_i((1+c_4\eps)\beta) \subset (1-\eps)R_i \subset R_i((1+c_3\eps)\beta).
  \]
\end{lemma}
\begin{proof}
  Fix a transmitter $s_i$.  Fix a point $q$ on the boundary of $R_i$.  Its SIN~ratio $E(q)$ relative to $s_i$
  is precisely $\beta$, since $R_i$ is convex
  \cite{aeklpr-sdciawn-12}.  Without loss of generality, $s_i$ lies at
  the origin and $q$ on the positive $x$-axis at distance $d$ from it.
  Put $d'=(1-\eps)d$.  Consider the point $q'=(d',0)$ on the boundary
  of $(1-\eps)R_i$.  We measure how different the SIN~ratio $E(q')$ with respect to $s_i$ is
  from $E(q)$ by bounding $E(q')/E(q)$. 
  As before, let $N(q)$ and $D(q)$ denote the numerator and denominator of $E(q)$, respectively.
  Then, $E(q')/E(q) = (N(q')/N(q)) \cdot (D(q)/D(q'))$, and
  $N(q')/N(q) = (1/d')^\alpha/(1/d)^\alpha=(d/d')^\alpha = (1-\eps)^{-\alpha}$.
  We now consider the ratio of the denominators (i.e., $D(q)/D(q')$).
  
 As before, the extreme values of the ratio of the denominators can be computed by setting $N=0$; one of the extremes corresponds to the denominator shrinking as we move from $q$ to $q'$ and the other to it increasing.

The maximum possible \emph{increase} in the denominator
corresponds to the maximum \emph{decrease} in the distances to all sites other than $s_i$,
which in turn corresponds to every such site $s_j$ lying on the negative $x$-axis (we will see below that $s_j$ cannot lie on the segment $s_iq$).  
Thus,
  \[
     \frac{D(q)}{D(q')} = \frac{\Sigma_{j \ne i}{\frac{1}{|q-s_j|^\alpha}}}{\Sigma_{j \ne i}{\frac{1}{|q'-s_j|^\alpha}}} \ge
     \min_{j \ne i}\frac{\frac{1}{|q-s_j|^\alpha}}{\frac{1}{|q'-s_j|^\alpha}} = \min_{j \ne i}(\frac{|s_j-q'|}{|s_j-q|})^\alpha =
     \min_{j \ne i}(\frac{|s_j-q|-|q'-q|}{|s_j-q|})^\alpha \ .
  \]
  As in the proof of Lemma~\ref{lem:tedious}, we note that 
  $E(q)=\beta>1$ implies $|s_i-q| \beta^{1/\alpha} \leq |s_j-q|$.\footnote{%
    Notice that this implies $s_j \not\in qs_i$, as promised.}
  Therefore, for all $j \ne i$,
  \[
    \frac{|s_j-q|-|q'-q|}{|s_j-q|} \geq \frac{\beta^{1/\alpha} |s_i-q|-\eps|s_i-q|}{\beta^{1/\alpha}
      |s_i-q|} =
    {1-\eps\beta^{-1/\alpha}},
  \]
  and
  \begin{align*}
    \frac{E(q')}{E(q)} & \ge 
    (1-\eps)^{-\alpha}(1-\eps\beta^{-1/\alpha})^{\alpha}
    = (\frac{1-\eps\beta^{-1/\alpha}}{1-\eps})^\alpha 
    = (\frac{1-\eps + (\eps-\beta^{-1/\alpha}\eps)} {1-\eps})^\alpha\\
    &
      = (1 +  \frac { (1  - \beta^{-1/\alpha})\eps } { 1 - \eps})^\alpha
      \ge  (1 + (1  - \beta^{-1/\alpha })\eps)^\alpha
      \ge  1 + (1  - \beta^{-1/\alpha })\eps,
  \end{align*}
  provided $\eps < 1$ and $\alpha\geq 1$.

  At the other extreme, the distance to $s_j$ is maximally increased if it lies on the positive $x$-axis beyond $q$.
So we assume that this is the case for every site $s_j$, $j\neq i$. Thus, 
\[
     \frac{D(q)}{D(q')} \le
     \max_{j \ne i}(\frac{|s_j-q'|}{|s_j-q|})^\alpha =
     \max_{j \ne i}(\frac{|s_j-q|+|q'-q|}{|s_j-q|})^\alpha \ .
  \]  
 Once again using $|s_i-q| \beta^{1/\alpha} \leq |s_j-q|$, we obtain, for all $j \ne i$,
  \[
    \frac{|s_j-q|+|q'-q|}{|s_j-q|}
    \leq \frac{\beta^{1/\alpha} |s_i-q|+\eps|s_i-q|}{\beta^{1/\alpha} |s_i-q|}
    = 1+\eps\beta^{-1/\alpha},
  \]
  yielding
  \begin{align*}
    \frac{E(q')}{E(q)} & \le
    \frac{(1+\eps\beta^{-1/\alpha})^{\alpha}}{(1-\eps)^{\alpha}} 
    = \frac{(1-\eps+(\eps+\eps\beta^{-1/\alpha}))^{\alpha}}{(1-\eps)^{\alpha}} \\
    & = (\frac{1-\eps+(1+\beta^{-1/\alpha})\eps}{1-\eps})^{\alpha} 
    = (1 + \frac{1+\beta^{-1/\alpha}}{1-\eps}\eps)^{\alpha} \\
    & \leq (1 + 2(1+\beta^{-1/\alpha})\eps)^{\alpha} \leq 
      1+2((2+\beta^{-1/\alpha})^\alpha-1)\eps,
  \end{align*}
  where we used $0<\eps<1/2$ and convexity of $x^\alpha$. 

  Now recall $R_i$ and $(1-\eps)R_i$ are star-shaped with center $s_i$, and $q$ (respectively, $q'$) is a boundary point of $R_i$ (respectively, $(1-\eps)R_i$).  We have therefore proven that 
  \[
    R_i((1+c_4\eps)\beta) \subset (1-\eps)R_i \subset R_i((1+c_3\eps)\beta),
  \]
  with $c_3=1-\beta^{-1/\alpha}$ and $c_4=2((2+\beta^{-1/\alpha})^\alpha-1)$, as promised.
\end{proof}

Notice that Lemmas~\ref{lem:tedious} and~\ref{lem:tedious2} imply that (up to a factor that depends only on $\beta>1$ and $\alpha$), the two notions of approximating $R_i$ defined in the introduction are actually equivalent, at least for the uniform power case.
Roughly, being off by a factor of $1\pm\delta$ in the scale of $R_i$ (which is the same as being off by a factor of $(1\pm\delta)^2 = 1\pm\Theta(\delta)$, when $\delta\leq 1$, in its area), corresponds to being off by a factor $1\mp C\delta$ in the SINR ratio $E(q)$ for the query point $q$, where $C$ depends on the constants $c_1$, $c_2$, $c_3$, and $c_4$ from the above lemmas, which in turn depend only on $\alpha$ and $\beta$.

\section{Speeding up preprocessing for a previous point-location result}
\label{sec:speed-up}

In this section we explain how Theorem~\ref{th:2d-general-uniform-ptas} together with the observations from Section~\ref{sec:structural} can be applied to speed up the preprocessing stage of an existing point-location algorithm.  (Similar techniques can be used 
in other situations, as long as they involve a large number of independent SIN ratio queries and can be modified to tolerate an approximate answer.)

Specifically, consider the data structure presented by Avin et al. \cite{aeklpr-sdciawn-12} for a set of $n$~uniform-power transmitters, with construction time%
\footnote{%
Again, the actual preprocessing time seems to be $O(n^2(\log n + 1/\delta))$, see below.}
$O(n^2/\delta)$ and query time $O(\log n)$, where $\delta>0$ is a given approximation parameter; the query returns \textsc{yes/no/maybe}, where \textsc{yes/no} are guaranteed to be correct and \textsc{maybe} occurs rarely in the sense that (the regions corresponding to each of the \textsc{yes/no/maybe} answers have well-defined area, see below, and) the area of the locus of points where the data structure returns \textsc{maybe} is at most $\delta \cdot \area(R_i)$.

The structure of \cite{aeklpr-sdciawn-12} is actually a collection of $n$ data structures $\mbox{DS}_i$, one per transmitter.
The data structure $\mbox{DS}_i$ for transmitter $s_i$ consists of an inner ($R_i^+$) and outer ($R_i^-$) approximation for reception region $R_i$, so that $\area(R_i^- \setminus R_i^+) \le \delta \cdot \area(R_i)$, see the definitions in the introduction.  
The construction of $\mbox{DS}_i$ uses the fact that 
the region $R_i$ is convex and fat; the latter in this case means that the ratio of the radius of the smallest disk enclosing $R_i$ and centered at $s_i$ and the radius of the largest disk enclosed in $R_i$ and centered at $s_i$ is  bounded by an expression that depends only on $\alpha$ and $\beta$ and not on $n$ or the geometry of the input \cite{aeklpr-sdciawn-12}.
After building the Voronoi diagram of the transmitter sites in $O(n \log n)$ time \cite{bcko-cgaa-08},
the construction of $\mbox{DS}_i$ consists of two stages. In the first,
they compute an explicit estimate on the size of $R_i$, by applying an exponential-search-like procedure along the segment connecting $s_i$ to its nearest other transmitter $s_j$ in $\S$, where
each comparison is resolved by explicitly evaluating the SIN ratio at some point~$q \in s_is_j$ and comparing it to $\beta$, i.e., by an (\emph{exact}) \textsc{in/out} test. In the second stage, a grid of size roughly $1/\delta \times 1/\delta$, scaled to exactly cover the outer disk is laid, and, by performing $O(1/\delta)$ additional \textsc{in/out} tests, the sets $R_i^+$ and $R_i^-$ are obtained, as collections of grid cells.
This algorithm thus performs $\Theta(\log n + 1/\delta)$ exact \textsc{in/out} tests per transmitter, at a cost of $\Theta(n)$ arithmetic operations each; the high cost of each test is the bottleneck of the preprocessing step.

Conceptually, we speed up the above algorithm by constructing the $n$ individual data structures $\mbox{DS}_i$ in parallel, forming $O(\log n + 1/\delta^2)$ batches of $n$ queries each, and use Theorem~\ref{th:2d-general-uniform-ptas} to evaluate each batch in near-linear time.  However, our batch query algorithm is not exact, but approximate; for some queries, instead of \textsc{in} or \textsc{out}, it answers \textsc{maybe}.  Below we explain how to handle this complication.  We first give a more detailed description of the procedure used by Avin~et~al., and highlight the necessary modifications.  Along the way, we will need to remind the reader how queries are performed in \cite{aeklpr-sdciawn-12}, so as to make sure that the same query mechanism still applies after our changes.

\begin{enumerate}[(i)]
\item The Voronoi diagram of the transmitter locations is built in $O(n \log n)$ time and adds $O(\log n)$ time to the query, as already mentioned above.
It is used at the query stage to identify the unique candidate site $s_i=s(q)$ for the query $q$.  Both versions focus on deciding whether $q$ can hear $s_i$, approximately, by constructing the data structure $\mbox{DS}_i$ for handling those queries whose closest site is $s_i$, for every $i$.  Our description will  hereafter be confined to building and querying of $\mbox{DS}_i$.

\item They first estimate the extent of $R_i$ by, conceptually, approximating the radii of the smallest enclosing and the largest enclosed circles of $R_i$, centered at $s_i$.  Since $R_i$ is fat, however, it is enough to restrict the calculations to a single ray emanating from $s_i$, namely, the one connecting it to the nearest transmitter location $s_j$.  Without loss of generality, put $s_i=(0,0)$ and assume that this ray is the positive ray of the $x$-axis, emanating from~$s_i$.  

They first estimate the extent of $R_i$ by finding two values $d_i, D_i$, $0 < d_i < D_i < |s_i-s_j|$, so that $(d_i,0) \in R_i$ and $(D_i,0) \not\in R_i$.  Theorem~4.1 in \cite{aeklpr-sdciawn-12} gives explicit expressions for the initial values of $d_i$ and $D_i$ and proves that $D_i/d_i$ is bounded by a polynomial in~$n$.  This calculation takes constant time (per transmitter), once the Voronoi diagram of the sites has been computed and nearest-neighbor information for each site extracted.  It does not involve computing the SIN~ratio for any point.

\item The next step refines the estimates on $d_i$ and $D_i$ %
to ensure $D_i/d_i=O(1)$.
In \cite{aeklpr-sdciawn-12}, this is done  by performing an exponential search, each time doubling the distance from $s_i$ to a point known to be in $R_i$ until a pair of consecutive points, one in $R_i$ and one outside of it, are obtained.  Since the initial estimate ratio is known to be polynomial in $n$, this requires $O(\log n)$ exact oracle calls (i.e., \textsc{in/out} tests), per transmitter.\footnote{%
    The stated total preprocessing time in \cite{aeklpr-sdciawn-12} seems to neglect the fact that this computation needs to be done for each transmitter, at the cost of $\Theta(n)$ per evaluation of the SIN ratio, for a total of $O(n^2 \log n)$; this term appears to be missing from the total running time in their analysis.}

We now describe how to perform the same procedure using an approximate oracle with precision parameter $\Delta$, to be fixed below.
We call the oracle on a sequence of points $(d_i\lambda^k,0)$, $k=1,2,\dots$, on $s_is_j$ whose distance from $s_i$ increases by a factor of $\lambda$ at each step, where $\lambda>1$ is a small constant. We artificially add $(d_i,0)$ (from step~(ii)) to the beginning of the sequence as a \textsc{yes} (even though the approximate oracle may not return a \textsc{yes} there) and $(D_i,0)$ (from step~(ii)) to the end of the sequence as a \textsc{no} (even though it may not return a \textsc{no}), to guarantee termination.

We return a pair of possibly non-consecutive points in this sequence so that the first is a \textsc{yes}, the last is a \textsc{no}, and all intervening points, if any, are \textsc{maybe}s.  By construction, the first point in our sequence is a \textsc{yes} and the last a \textsc{no}, hence the claimed pair exists.  We evaluate the entire sequence and return such a pair.  Notice that there cannot be two such pairs, as the sequence starts with a \textsc{yes} and a \textsc{no} cannot appear before a \textsc{yes}, for that would imply the existence of a point outside of $R_i$ between two points in it, contradicting the convexity of $R_i$.

We again refer to the resulting pair of values as $d_i$ and~$D_i$, respectively, although of course they may be different from the original $d_i, D_i$ pair found in step~(ii). Our goal now is to show that now $D_i/d_i = O(1)$. For this we need the following lemma whose proof is similar to the first part of the proof of Lemma~\ref{lem:tedious}.
 
\begin{lemma}
\label{lem:also_tedious}
Assume $\beta>1$, and
let $q$ be a point such that $E(q) \ge \beta/\rho$ with respect to~$s_i$, where $1 \le \rho < \beta$, and
let $q'$ be the point on the ray emanating from $s_i$ and passing through $q$ at distance $\tau |q-s_i|$ from $s_i$, $\tau > 1$.
Then 
\[
   \frac{E(q')}{E(q)} \le (1 - (1 - (\rho/\beta)^{1/\alpha})(1 - 1/\tau))^\alpha\ .
\]
\end{lemma}

\begin{proof}
  Without loss of generality, $s_i$ lies at
  the origin and $q$ and $q'$ on the positive $x$-axis at distances $d$ and $d'=\tau d$, respectively, from it.
	For a point $p$, let $N(p)$ and $D(p)$ denote the numerator and denominator of the expression for $E(p)$ with respect to $s_i$, respectively.
  Then, $E(q')/E(q) = (N(q')/N(q)) \cdot (D(q)/D(q'))$, and
  $N(q')/N(q) = (1/d')^\alpha/(1/d)^\alpha=(d/d')^\alpha = \tau^{-\alpha}$. 
  We now consider the ratio $D(q)/D(q')$ of the denominators.

  As noted in the proof of Lemma~\ref{lem:tedious},
	the maximum value of the ratio of the denominators is achieved when $N=0$ and the denominator $D()$ decreases as much as possible as we move from $q$ to $q'$, i.e, when the distances to all sites other than $s_i$ increase as much as possible as we move from $q$ to $q'$. This in turn corresponds to the case where each such site~$s_j$ lies on the negative $x$-axis (we will see below that $s_j$ cannot lie on the segment $s_iq$). Thus,
  \[
     \frac{D(q)}{D(q')} = \frac{\Sigma_{j \ne i}{\frac{1}{|q-s_j|^\alpha}}}{\Sigma_{j \ne i}{\frac{1}{|q'-s_j|^\alpha}}} \le
     \max_{j \ne i}\frac{\frac{1}{|q-s_j|^\alpha}}{\frac{1}{|q'-s_j|^\alpha}} = \max_{j \ne i}(\frac{|s_j-q'|}{|s_j-q|})^\alpha =
     \max_{j \ne i}(\frac{|s_j-q|+|q'-q|}{|s_j-q|})^\alpha \ .
  \]
  We now observe that
  $|s_i-q| (\beta/\rho)^{1/\alpha} \leq |s_j-q|$ for all $j \ne i$ (this follows from the fact that, since the SIN~ratio of $q$ is at least $\beta/\rho$, the signal strength of $s_i$ at $q$ is at least $\beta/\rho$ times that of $s_j$ at $q$, i.e., $1/|s_i-q|^\alpha \geq (\beta/\rho)/|s_j-q|^\alpha$).\footnote{
    Notice that this implies $s_j \not\in s_iq$.}
  Therefore, for all $j\neq i$,
  \[
    \frac{|s_j-q|+|q'-q|}{|s_j-q|} \leq \frac{(\beta/\rho)^{1/\alpha} |s_i-q|+(\tau-1)|s_i-q|}{(\beta/\rho)^{1/\alpha}
      |s_i-q|} =
    {1+(\tau-1)(\rho/\beta)^{1/\alpha}},
  \]
  and
  \begin{align*}
    \frac{E(q')}{E(q)} & = \tau^{-\alpha} \cdot \frac{D(q)}{D(q')} \le
    \tau^{-\alpha} \cdot (1+(\tau-1)(\rho/\beta)^{1/\alpha})^\alpha \\
    & = (\frac{1+(\tau-1)(\rho/\beta)^{1/\alpha}}{\tau})^\alpha 
    = (\frac{\tau - ((\tau-1)-(\rho/\beta)^{1/\alpha}(\tau-1))} {\tau})^\alpha\\
    &
      = (1 - (1 - (\rho/\beta)^{1/\alpha})(1 - 1/\tau))^\alpha\ .
      \qedhere
  \end{align*}
\end{proof}

We are now ready to show that $D_i/d_i = O(1)$.
First, so that we do not have to distinguish between two cases, if $D_i$ was not determined from the oracle calls on the points $(d_i\lambda^k,0)$, $k=1,2,\ldots$, i.e., if $D_i$ is still the value obtained in step~(ii), then we replace it by the first value above it of the form $d_i\lambda^k$. 
Next, we may assume that $D_i/d_i>\lambda^2$ (since otherwise we are done),
and therefore the points $q = (d_i\lambda,0)$ and $q' = (D_i/\lambda,0)$ are distinct.

We upper bound $D_i/d_i$ by bounding the ratio $E(q')/E(q)$ both from below and from above.
We know that our approximate oracle returned the answer \textsc{maybe} for both these points. Therefore, 
since it did not return~\textsc{yes} for $q$, we may conclude that $\EE(q) < (1+\Delta)\beta$, and since it did not return~\textsc{no} for $q'$, we may conclude that $\EE(q') > (1-\Delta)\beta$. This allows us to lower bound the ratio $E(q')/E(q)$:
\[
\frac{E(q')}{E(q)} \ge \frac{\EE(q')/(1+\Delta)}{\EE(q)/(1-\Delta)} \ge (\frac{1-\Delta}{1+\Delta})^2\ .
\]

We now apply Lemma~\ref{lem:also_tedious} to upper bound the ratio $E(q')/E(q)$.
Since our approximate oracle did not return the answer \textsc{no} for $q$, we may conclude that $\EE(q) > (1-\Delta)\beta$ and set $\rho=1+\Delta$ in Lemma~\ref{lem:also_tedious}. Moreover, since $|q'-s_i| = (D_i/(d_i \lambda^2))|q-s_i|$, we set $\tau=D_i/(d_i \lambda^2)$ in Lemma~\ref{lem:also_tedious}, to obtain
\[
 \frac{E(q')}{E(q)} \le (1 - (1 - (\frac{1+\Delta}{\beta})^{1/\alpha})(1 - \frac{d_i\lambda^2}{D_i}))^\alpha\ .
\] 

Combining the two estimates, we obtain
\[
(\frac{1-\Delta}{1+\Delta})^2 \le (1 - (1 - (\frac{1+\Delta}{\beta})^{1/\alpha})(1 - \frac{d_i\lambda^2}{D_i}))^\alpha\ ,
\]
which after some manipulation yields
\[
\frac{D_i}{d_i} \le \lambda^2 \cdot \frac{1 - (\frac{1+\Delta}{\beta})^{1/\alpha}}{(\frac{1-\Delta}{1+\Delta})^{2/\alpha} - (\frac{1+\Delta}{\beta})^{1/\alpha}}\ ,
\]
provided both the denominator and the numerator of the last fraction are positive,
which introduces the restriction $\beta \ge \frac{(1+
\Delta)^3}{(1-\Delta)^2}$, which in turn is less strict than the requirement $\beta \ge 1+6\Delta$ (assuming $\Delta \le 1/14$).

\item At this point, the authors of \cite{aeklpr-sdciawn-12} have obtained a constant-factor approximation of $D_i$.  Then they construct a square just large enough to cover a disc of radius $D_i$ times the maximum fatness allowed by Corollary~5.4 from \cite{aeklpr-sdciawn-12}, centered at $s_i$, and therefore fully containing $R_i$; 
since $R_i$ is fat, it occupies a constant fraction of 
the square.  This square is subdivided into a~$k \times k$~grid, where $k=\Theta(1/\delta)$ and $\delta$ is the desired approximation parameter.  Avin et al. \cite{aeklpr-sdciawn-12} now classify the grid squares as \emph{internal}, \emph{external}, or \emph{mixed}, depending on whether all/none/some (but not all or none) of its corners are contained in $R_i$, respectively.  Conceptually, in their construction, the inner approximation $R^+_i$ of $R_i$ consists of the internal squares only, mixed squares intersect the boundary of $R_i$, and external squares may or may not intersect $R_i$, so an additional test is needed.  They show \cite[Lemma~5.3 and Corollary~5.4]{aeklpr-sdciawn-12} 
that any external cell more than $C$ grid cells (for a constant $C=C(\alpha,\beta)$) away from any mixed cell must lie entirely outside of $R_i$.  Therefore, an outer approximation $R_i^-$ can be obtained by collecting all internal and mixed cells, and all the external cells within $C$ cells of any mixed cell.

Identification of internal/mixed/external cells can be performed by brute force, evaluating the SIN ratio at every one of the $k^2=\Theta(1/\delta^2)$ grid nodes.  The algorithm in \cite{aeklpr-sdciawn-12} is cleverer in that it only examines $\Theta(k)$ cells by taking advantage of the fact that $R_i$ is convex, speeding up preprocessing to only $\Theta(k)=\Theta(1/\delta)$ oracle calls per transmitter, for a total of $O(n^2/\delta)$ work.

Since both $R_i^+$ and $R_i^-$ are represented as (unions of) collections of grid squares, such that the squares in any one row are contiguous, they can be stored in an array of size $O(1/\delta)$, with two entries for each grid row, identifying the cells where each set starts and ends.  Thus (this portion of) the query~$q$ can be performed in constant time, provided \emph{floor} operation is available, by identifying the grid row containing~$q$, followed by a look-up in the above array, to verify whether or not $q$ belongs to the range of cells covered by $R_i^+$ and/or $R_i^-$.

This completes the description of the preprocessing and query algorithms as presented in \cite{aeklpr-sdciawn-12}.
In our adaptation of this step of their algorithm, we encounter several difficulties.  The first one is that our oracle only provides approximate answers.  The second is that, even though $R_i$ itself is convex, we cannot guarantee that the region where the oracle of Theorem~\ref{th:2d-general-uniform-ptas} provides, say, \textsc{yes}~answers, is convex.
Therefore we will not use the cleverer grid-traversal idea of Avin et al. \cite{aeklpr-sdciawn-12} and will settle for just calling the approximate oracle with approximation parameter~$\Delta$, to be specified below, on every single grid node.  This requires $k^2$ oracle calls for each transmitter, but will be done in a batched manner, each batch consisting of one call for each transmitter. 
The total cost of all these queries will be $O(n\delta^{-2}\Delta^{-1/2} \log^4 n \log \log n)$, as per Theorem~\ref{th:2d-general-uniform-ptas}.

We now modify the internal/mixed/external classification of Avin et al.\ for the situation where we use an approximate oracle, as follows:
A grid cell where the \emph{approximate} oracle reports \textsc{yes} at all four corners is \emph{internal}.  Intuitively, the new inner approximation $R^{++}_i$ of $R_i$ is again defined as the set of all internal grid squares.  For technical reasons explained below, we will extend the definition slightly, as follows:  In every row of the grid, add to $R_i^{++}$ all squares from the leftmost internal square to the rightmost one, inclusively.  In other words, we also include possibly non-internal squares ``sandwiched'' between internal ones.

Observe that by convexity of $R_i$, both an internal square and a ``sandwiched'' non-internal square are fully contained in $R_i$.  Therefore $R_i^{++} \subset R_i$.\footnote{It is not difficult to see that, in fact, $R_i^{++} \subset R_i^+$, 
though we will not use this fact below.}

We claim that $\area(R_i \setminus R_i^{++}) \leq O(\delta) \cdot \area(R_i)$.  Consider a grid point $q \in R_i((1+3\Delta)\beta)$, so that $E(q)\geq (1+3\Delta)\beta$.  By the approximation guarantee of the oracle, $\EE(q)\geq (1-\Delta)(1+3\Delta)\beta \geq (1+\Delta)\beta$, provided $0 < \Delta \leq 1/3$, and the approximate oracle returns \textsc{yes} on~$q$.  Hence every grid square all of whose corners lie in $R_i((1+3\Delta)\beta)$ is contained in $R_i^{++}$.  We cannot conclude that $R_i((1+3\Delta)\beta) \subseteq R_i^{++}$, but we may conclude that
$R_i((1+3\Delta)\beta) \setminus R_i^{++}$ is fully contained in the union of grid squares intersected by the boundary of $R_i((1+3\Delta)\beta)$.  The latter set is convex, so its boundary intersects the grid lines at most $4k$ times and therefore meets at most $4k$ grid squares.  In particular,
$\area(R_i((1+3\Delta)\beta) \setminus R_i^{++}) \leq O(\delta) \cdot \area(R_i)$, and therefore, 
$\area(R_i \setminus R^{++}_i) \leq \area(R_i \setminus R_i((1+3\Delta)\beta)) + O(\delta) \cdot \area(R_i)$.

Now observe that, by Lemma~\ref{lem:tedious2}, $(1-3\Delta/c_3)R_i \subset R_i((1+3\Delta)\beta)$, where $c_3=c_3(\alpha,\beta)$ is the constant from the lemma, as long as $0<3\Delta/c_3 < 1/2$, or $0<\Delta<c_3/6$.
So, $\area(R_i \setminus R^{++}_i) \leq \area(R_i \setminus (1-3\Delta/c_3)R_i) + O(\delta) \cdot \area(R_i)$.

Note that
\begin{align*}
  \area (R_i \setminus (1-3\Delta/c_3)R_i) 
& = \area (R_i) - \area ((1-3\Delta/c_3)R_i) \\
& = \area (R_i) - (1-3\Delta/c_3)^2 \area (R_i) \\
& = (6\Delta/c_3 - 9\Delta^2/c_3^2) \area (R_i)\\
& = \Theta(\Delta) \area(R_i),
\end{align*}
for the range of $\Delta$ under consideration.

Therefore we may conclude that
\[
\area(R_i \setminus R^{++}_i) \leq \Theta(\Delta)\area(R_i) + O(\delta)\area(R_i) = O(\delta)\area(R_i),
\]
\old{
Therefore we may conclude that\footnote{%
  We use the following identity: For four sets $A$, $B$, $C$, $D$ with $B,C \subset D \subset  A$, $A \setminus B \subset (A \setminus C) \cup (D \setminus B)$.}
\begin{align*}
\area(R_i \setminus R^{++}_i) 
& \leq
\area (R_i \setminus (1-3\Delta/c_3)R_i) +
\area (R_i((1+3\Delta)\beta) \setminus R_i^{++})\\
& = (\Theta(\Delta) + O(\delta))\area(R_i) = O(\delta)\area(R_i),
\end{align*}
}
provided we pick $\Delta=O(\delta)$.

We now construct the upper approximation $R^{--}_i$ to $R_i$.
An \emph{external} grid square is one where the approximate oracle returns \textsc{no} on all corners.  Grid squares that are not internal or external are \emph{mixed}. 
Let $X$ be the set of internal and mixed grid squares, i.e., the squares where the approximate oracle did not return \textsc{no} on all corners.  We define $R^{--}_i$ as the union of the cells in $X$ plus all external grid cells within $C+\sqrt2$ grid cells of a mixed cell in $X$, i.e., the external cells whose Euclidean distance to a mixed cell is at most $C+\sqrt2$~times the edge length of a grid cell.  Once again, to speed up the query (see below), in every grid row, we also add to $R^{--}_i$ all the cells ``sandwiched'' between the leftmost and the rightmost identified by the above rules, even if they don't satisfy the rules themselves.

By the approximation guarantee of the oracle, an external grid cell is also external relative to the exact oracle.
Notice that the set of mixed cells according to the approximate oracle includes all the mixed cells according to the exact one.  
Therefore, the set of cells our algorithm includes in $R_i^{--}$ is a superset of those included in $R_i^-$, so  $R_i^{--}$ contains $R_i$.  (Note that including the additional ``sandwiched'' cells only enlarges $R^{--}_i$ and therefore does not affect this statement.) 
We now estimate $\area(R_i^{--}\setminus R_i)$.

Observe that $R_i((1-2\Delta)\beta) \subset (1+2\Delta/c_1)R_i$, by Lemma~\ref{lem:tedious}, provided $0 < 2 \Delta/c_1 < 1/c_1$ or $0 < \Delta < 1/2$.

Consider a corner $q$ of a grid cell in~$X$ for which the approximate oracle did not return a~\textsc{no}. By definition, $\EE(q) \geq (1 - \Delta)\beta$ and therefore $E(q) \geq (1-\Delta)\EE(q) \geq (1 - \Delta)^2 \beta > (1 - 2\Delta)\beta$. Hence, every such corner lies in $R_i((1-2\Delta)\beta)$ and therefore in $(1+2\Delta/c_1)R_i$.
Therefore, by convexity of $(1+2\Delta/c_1)R_i$, any square in $X$ without \textsc{no} corners is contained in it.  A square of $X$ that is not fully contained in $(1+2\Delta/c_1)R_i$ must necessarily have at least one corner outside of $(1+2\Delta/c_1)R_i$ and at least one inside it.  Therefore, such a square must meet the boundary of $(1+2\Delta/c_1)R_i$.  Hence any cell of $X$ lies either in $(1+2\Delta/c_1)R_i$ or in $(1+2\Delta/c_1)R_i$ Minkowski expanded by a disk of radius $\sqrt2$ grid cells.  In particular, all cells of $R^{--}_i$ are contained in $(1+2\Delta/c_1)R_i$ expanded by a disk of radius $C+2\sqrt2$ grid cells.  (Notice that this also applies to the ``sandwiched'' cells, as the latter set is convex.)

Recall that $\area ((1+2\Delta/c_1)R_i) = (1+2\Delta/c_1)^2\area(R_i)$ and therefore $\area ((1+2\Delta/c_1)R_i \setminus R_i) = (4\Delta/c_1+4\Delta^2/c_1^2)\area (R_i) = \Theta(\Delta) \area(R_i)$.  We now estimate the area gained by expanding $(1+2\Delta/c_1)R_i$ by the disk of radius $C+2\sqrt2$ grid cells. 
(The additional $\sqrt2$ appears because a grid cell lying at a distance $x$ from a set lies fully in an expansion of the set by a disk of radius $x+\sqrt2$.)

Invoking Steiner formula,\footnote{%
  Steiner formula states that $\area(A \oplus tB) = \area(A) + t \cdot \mathop{\mathit{perimeter}}(B) + t^2 \cdot \area(B)$, if $A$ is convex, $B$ is a unit circle, and $tB$ is a circle of radius $t$, in the plane; $\oplus$ denotes Minkowski addition of sets.}
 the area of the expanded set, in grid squares, is the area of $(1+2\Delta/c_1)R_i$ plus the area of the disk ($\pi (C+2\sqrt2)^2=O(1)$ grid cells) plus the radius of the latter times the perimeter of the former, which are  $(C+2\sqrt2)=O(1)$ and $\Theta(k)$ respectively.  
Since the expanded set contains all the grid cells of $R^{--}_i$, we 
conclude that it gains at most $\Theta(k)$ additional grid squares over its unexpanded version, i.e., area $\Theta(\delta)\area(R_i)$, as claimed.

Thus 
\begin{align*}
  \area(R^{--}_i \setminus R_i) & \leq \area((1+2\Delta/c_1)R_i\setminus R_i) + \Theta(\delta)\area(R_i) \\
  & \leq (\Theta(\Delta)+\Theta(\delta))\area(R_i)\\
  & = \Theta(\delta) \area(R_i),
\end{align*}
provided $\Delta = O(\delta)$.
\end{enumerate}

Therefore, $\area(R^{--}_i \setminus R^{++}_i) = \area(R^{--}_i \setminus R_i) + \area(R_i \setminus R^{++}_i)  = O(\delta) \area(R_i)$, as claimed, for a sufficiently small $\Delta=O(\delta)$,
concluding the proof of the approximation correctness for our algorithm for preprocessing $R_i$ for point location.

Moreover, since both $R_i^{++}$ and $R_i^{--}$ contain only sets of contiguous squares (if any at all) in every grid row, a point can be located in each set in $O(1)$ time provided \emph{floor} operation is available, by identifying the row containing the point, using its $y$-coordinate and then comparing its $x$-coordinate to that of the leftmost and rightmost grid cell in the row.

The following theorem summarizes the main result of this section:
\begin{theorem}
Fix positive integer $\alpha$ and $\beta>1$.
Given a set $\S$ of $n$ transmitters (all of power~1) and a parameter $\delta > 0$,
we can do the following in total time $O(n \delta^{-2.5} \log^4 n \log \log n)$.\footnote{Our previous discussion implies that one can set the approximation parameter $\Delta$ for the approximate oracle to $\Delta = c\delta$, with 
$c=c(\alpha,\beta)$ a sufficiently small constant.  In fact, 
setting $\Delta = \min\{\delta, (\beta-1)/6, (1-\beta^{-1/\alpha})/6, 1/14\}$ satisfies all the requirements.}
Compute, for each transmitter $s_i$, two sets $R_i^{++}$ and $R_i^{--}$ (each represented as a union of a collection of grid cells), such that
(i)~$R_i^{++} \subseteq R_i \subseteq R_i^{--}$ and (ii)~$\area(R^{--}_i \setminus R^{++}_i) = O(\delta) \area(R_i)$.
Moreover, given a query point $q$ and $R_i^{--}$ and $R_i^{++}$, we can determine in $O(1)$ time whether $q$ is in $R_i^{++}$, $R_i^{--} \setminus R_i^{++}$, or $\RR \setminus R_i^{--}$.  Given $q$, we can thus approximately determine if it can receive any transmitter, and if so, which one, in $O(\log n)$ time. 
\end{theorem}

\section{Concluding remarks}
\label{sec:conl}

We described several algorithms that combine computational geometry techniques and methods of computer algebra to obtain very fast batched SINR diagram point-location queries.

Besides speeding up the construction time of known structures, we would like to identify further applications of batched point location to other problems studied in the SINR model.

We note that our results are general, in the sense that analogous results can be obtained for diagrams that are induced by other inequalities similar to the SINR inequality. 

Finally, on a larger scale, we are interested in other instances where algebraic and geometric tools can be combined to achieve significant results that are seemingly impossible without it.

\paragraph*{Acknowledgments.}
The authors would like to acknowledge significant help of Sariel Har-Peled in matters of approximation and tremendous assistance of Guillaume Moroz in matters of computer algebra.  B.A.\ would also like to thank Pankaj K.\ Agarwal for general encouragement and moral support.

\old{
\section{Other problems that we discussed}

\begin{enumerate}
\item
If we can do batched point location in the plane, then we may be able to improve the running time of the known algorithms for the maximum capacity problem, which is defined as follows.
Let $\L = \{(c_1, s_1), \ldots, (c_n, s_n)\}$ be a set of $n$ pairs of points in the
plane representing $n$ (directional) links. In a link $(c_i, s_i) \in \L$, the point $c_i$
represents the receiver and the point $s_i$ represents the sender.
The goal is to find a maximum cardinality subset
of links that can operate simultaneously in the SINR model.

\item
Can we say something on the following problem which we call
{\em maximum capacity with acknowledgment} and that seems to be new.
Find a maximum subset $\L' \subseteq \L$ of links that can operate simultaneously as follows.
In the first part of a round
all senders in $\L'$ send their messages simultaneously to their corresponding receivers, and in the second part
all receivers in $\L'$ acknowledge receipt of their messages simultaneously.

\item
We also tried to come up with some interesting problem, where we are given a set of transmitters $\S$ and a set of receivers $\Q$, but not the transmitter-to-receiver matching.
   
\end{enumerate}
}

\newpage
\bibliographystyle{abbrv}
\bibliography{sinr}

\newpage
\appendix 

\section{Tools and definitions from the world of computer algebra}
In this section we state several well-known results from computer algebra and refer to \cite{gg-mca-99} for details; see \cite{bp-pmcfa-94} for an alternative presentation.  

\subsection{Definitions}
 
A \emph{univariate polynomial $A(x)$} of degree at most $n$ is defined
by an expression of the form $A(x)=a_0+a_1x+a_2x^2+\dots+a_nx^n$.  The
tuple 
$C(A) \coloneqq \langle a_0,a_1,\dots,a_n\rangle$ 
is the \emph{coefficient representation} of $A$.
Given a set $X=\{x_0,\dots,x_n \}$ of $n+1$ numbers, 
we let $V(A)=V(A,X)$ denote
$\{(x_0,A(x_0)), \dots, (x_n,A(x_n))\}$, a \emph{value representation} of $A$.

A \emph{fractional function} $F(x)$ of degree at most $n$ is a
function that can be written in the form $F(x)=A(x)/B(x)$ for two
polynomials $A$ and $B$ of degree at most $n$, with $B\not\equiv 0$.
\emph{Coefficient representation} of $F$ is simply $(C(A),C(B))$, the
pair of coefficient representations of its numerator and denominator.

Analogously, a \emph{bivariate polynomial $B(x,y)$} of degree at most~$n$ in each variable is defined by an expression of the form
$B(x,y)=\sum_{0 \leq i, j \leq n} b_{ij} x^iy^j$.  The set $
\{\langle i,j,b_{ij} \rangle\}$ is the \emph{coefficient representation} of
$B$.
The set $\{ (x_i,y_i,B(x_i,y_i)) \}$ is a \emph{value representation} of $B$, for a suitably large set of points $\{ (x_i,y_i) \} \in \RR^2$. 

One can analogously define bivariate fractional functions.

\subsection{Univariate tools and facts}
\label{sec:uni-tools}

The following two results can be found in Corollaries~10.8 and~10.12 in \cite{gg-mca-99}, respectively.

\begin{fact}[Univariate Multipoint Evaluation and Interpolation]
  \label{fact:eval-inter}
  If $A$ is a univariate polynomial of degree at most $n$ and $X$ a tuple of $n+1$ distinct numbers,  then the coefficient representation clearly determines $A$, but so does its value representation $V(A,X)$.
  \begin{description}
  \item[Multipoint Evaluation] Given the coefficient representation
    $C(A)$ of $A$, $A$ can be evaluated at the $n+1$ points of $X$,
    yielding $V(A,X)$ in $O(n \log^2 n \log \log n)$ arithmetic
    operations.
    If $X$ contains $m$ points, the computation can be done in $O((n+m) \log^2 n \log \log n)$ arithmetic operations.\footnote{%
      The second statement follows immediately from the first by
      extending~$X$ to length~$n$ with dummy data if $m<n$, or cutting
      it into chunks of length $n$ if $m>n$.}
  \item[Interpolation] Given a value representation $V(A,X)$ of a
    polynomial~$A$ of degree at most~$n$ on a set~$X$ of~$n+1$ points,
    the coefficient representation $C(A)$ of $A$, can be constructed
    using $O(n \log^2 n \log \log n)$ arithmetic operations.
  \end{description}
\end{fact}

Given two univariate polynomials $A$ and $B$ of degree at most $n$
each, let their \emph{product} $AB$ be the polynomial $D$ defined by $D(x)=A(x)B(x)$ for all $x$.  Multiplying two polynomials in
value representation is easy: for each $x_i \in X$, $D(x_i)=A(x_i)\cdot B(x_i)$,
by definition.  Somewhat surprisingly, one can also quickly multiply univariate polynomials in coefficient representation (see Theorem~8.23 in \cite{gg-mca-99}):

\begin{fact}[Fast Multiplication of Univariate Polynomials]
  \label{fact:fast-multiply}
  Given two univariate polynomials $A$, $B$ of degree at most $n$
  each, in coefficient representation, one can construct their product
  $AB$, using $O(n \log n \log \log n)$ arithmetic operations.
\end{fact}

\begin{note}
  The Fast Fourier Transform (FFT) is essentially an evaluation of a
  polynomial on a special set of complex numbers, the $n$th roots of
  unity (more precisely, the $2^{\lceil \log_2 n \rceil}$th roots of
  unity).  It can be performed in $O(n \log n \log \log n)$
  operations, due to the special structure of the set of roots of
  unity; technically, the $\log \log n$ term appearing in many of our
  bounds is due to the assumption that the appropriate primitive root
  of unity is not available explicitly; otherwise the bounds improve
  by a factor of $\log \log n$.  The Inverse FFT reverses the process
  (by using a variant of the FFT code), reconstructing a polynomial
  from its values at the roots of unity, again in $O(n \log n \log
  \log n)$ operations. The fast polynomial multiplication algorithm is an FFT,
  followed by point-wise multiplication, followed by the Inverse FFT.
\end{note}

The following observation has been made by previous authors; we provide a
proof for completeness; we follow the construction from
\cite{ma-cdplbf-12,ma-cdplbf-16}. 

\begin{lemma}[Sum of Fractions \cite{ma-cdplbf-12,ma-cdplbf-16}]
  \label{lemma:sum-fractions}
  Given a set of $n$ fractional functions $P_i(x)/Q_i(x)$ of constant
  degree each, in coefficient representation, the coefficient
  representation of their sum can be constructed using $O(n \log^2 n
  \log \log n)$ arithmetic operations.
\end{lemma}

\begin{proof}
  Given two fractions $\frac A B$ and $\frac C D$, their sum can be
  written as $\frac{AD+BC}{BD}$.  Therefore, given two fractions of
  degree at most~$d$ in coefficient representation, we can obtain their sum,
  which is a fraction of degree at most~$2d$, in coefficient
  representation, using three polynomial multiplications and one
  addition, for a total of $O(d \log d \log \log d)$ arithmetic
  operations.

  We start with the initial fractions $P_i/Q_i$, $i=1,\dots,n$ and add
  them in pairs, then add the resulting sums in pairs, and so forth.
  One round costs $O(n \log n \log \log n)$ operations, since the sum
  of the degrees of intermediate polynomials at each level is~$O(n)$.
  In each round, the number of fractions reduces by a factor of two,
  so the procedure stops after $O(\log n)$ rounds.  Hence the final
  result will be obtained after $O(n \log^2 n \log \log n)$
  operations.
\end{proof}

\begin{corollary}[Evaluation of Sum of Fractions]
  \label{cor:eval-sum-fractions}
  Given a set of $n$ fractional functions $P_i(x)/Q_i(x)$ of constant
  degree each, in coefficient representation, and a set of $m$ values
  $x_j$, one can compute the $m$ values $\sum_i P_i(x_j)/Q_i(x_j)$ in 
  time $O((n + m) \log^2 n \log \log n)$.
\end{corollary}

\begin{proof}
  Symbolically sum the fractions to construct one degree-$O(n)$
  fraction, using Lemma~\ref{lemma:sum-fractions}, evaluate the
  numerator and denominator at each $x_j$ using
  Fact~\ref{fact:eval-inter}, and divide.
\end{proof}

\subsection{Bivariate tools and facts}
\label{sec:bivariate-tools}

It would be helpful to have analogous tools for the bivariate case.
The difficulty is that a general bivariate polynomial of degree at most~$n$ in
each variable is described by $(n+1)^2 = \Theta(n^2)$ coefficients, so its
explicit coefficient representation necessarily has size $\Theta(n^2)$.  Thus
in general one cannot expect that direct analogs of
Facts \ref{fact:eval-inter} and \ref{fact:fast-multiply}, and
Lemma~\ref{lemma:sum-fractions} exist, at least not in the same form.

We define the \emph{product} $D=A B$ of two bivariate polynomials
$A$ and $B$ similarly to that of univariate ones.  Once again,
multiplying two polynomials in value form is easy by multiplying the
corresponding values.  Somewhat surprisingly, one can also quickly
multiply bivariate polynomials in coefficient representation:

\begin{fact}[Fast Multiplication of Bivariate Polynomials; \cite{p-smpm-94,bp-pmcfa-94} or {\cite[Corollary 8.28]{gg-mca-99}}]
  \label{fact:multiply-binomial}
  Given two bivariate polynomials $A(x,y)$, $B(x,y)$ of degree at most
  $n$ in each variable, in coefficient representation (which has size
  $\Theta(n^2)$ in general), one can construct their product $AB$,
  using $O(n^2 \log n \log \log n)$ arithmetic operations.
\end{fact}

Notice that, once we have the multiplication result, one can use the
same reasoning as in Lemma~\ref{lemma:sum-fractions} to obtain the
following:

\begin{fact}[Sum of Bivariate Fractions]
  \label{fact:sum-fractions-bi}
  Given a set of $n$ bivariate fractional functions
  $P_i(x,y)/Q_i(x,y)$ of constant degree each, in coefficient
  representation, the coefficient representation of their sum
  $P(x,y)/Q(x,y)$, can be constructed using $O(n^2 \log n \log \log n)$
  arithmetic operations.
\end{fact}

\begin{proof}
  For completeness, we outline the analysis. 
  We repeat the reasoning from the proof of Lemma~\ref{lemma:sum-fractions}, using Fact~\ref{fact:multiply-binomial} for bivariate multiplication.
  Round~$i$ involves polynomials of degree $O(2^i)$ with degrees summing to $O(n)$, so the number of arithmetic operations used in round~$i$ is $O(2^i n \log n \log \log n)$.  This sums up to $O(n^2 \log n \log \log n)$ for the $\log n$~rounds, as claimed.
\end{proof}

\begin{fact}[Bivariate Evaluation and Interpolation over a Grid \cite{p-smpm-94,bp-pmcfa-94}]
  \label{fact:bi:eval-inter-grid}
  Consider a bivariate polynomial $A(x,y)$ of degree at most~$n$ in
  each variable, and two sets $X=\{x_i\}$ and $Y=\{y_j\}$ of $n$
  numbers each.
  \begin{description}
  \item[Grid Evaluation]
    Given the $O(n^2)$ coefficients of $A$, $A$ can be evaluated at the
    points of $X\times Y$ (i.e., values $A(x_i,y_j)$ computed, for all
    combinations of $i$ and $j$) in time $O(n^2 \log^2 n \log \log n)$.
  \item[Interpolation from Values on the Grid]
    Given $(n+1)^2$ values $v_{ij}=A(x_i,y_j)$, one can reconstruct the
    coefficient representation of $A$ in time $O(n^2 \log^2 n \log \log n)$.
  \end{description}
\end{fact}

\begin{note}
  We outline one way to achieve the above bounds: Write $A(x,y)=\sum_{j=0}^{n}B_j(x)y^j$, where $B_j$ is a
  univariate polynomial of degree at most $n$ and use the
  univariate algorithm from Fact~\ref{fact:eval-inter} to
  multi-evaluate $B_j$ on all values $x_i$; the remaining arithmetic
  can be done in $O(n^2)$ operations; this results in an algorithm with $O(n^2 \log^2 n \log \log n)$ arithmetic operations.

  Interpolation can be accomplished by invoking univariate tools as well.  Treating $A(x_i,y)$ as a univariate polynomial in $y$, for every $x_i$ separately, and applying univariate interpolation (Fact~\ref{fact:eval-inter}) allows us to write it as $\sum B_j(x_i)y^j$, where $B_j(x)$ is as above and we have obtained the values $B_j(x_i)$ for all $i,j$. We now reconstruct each $B_j$ by using univariate interpolation again.  Once we are done, we have all the coefficients of all the $B_j$'s, which are the coefficients of $A$.

  This requires $2n$ invocations of univariate interpolation for a total of $O(n^2 \log^2 n \log \log n)$ arithmetic operations, as claimed.

\end{note}

\begin{corollary}[Evaluation of Sum of Fractions]
  \label{cor:eval-sum-fractions-bi}
  Given a set of $n$ bivariate fractional functions
  $P_i(x,y)/Q_i(x,y)$ of constant degree each, in coefficient
  representation and two sets $X=\{x_j\}$, $Y=\{y_k\}$ of $n$ values
  each, one can compute the $n\times n$ values $\sum_i
  P_i(x_j,y_k)/Q_i(x_j,y_k)$ in time $O(n^2 \log^2 n \log \log n)$.
\end{corollary}

\begin{proof}
  Combine the fractions into a single fraction using
  Fact~\ref{fact:sum-fractions-bi} and then evaluate it on the grid
  $X\times Y$ using Fact~\ref{fact:bi:eval-inter-grid}.
\end{proof}

There does not seem to be an easy analogue of fast univariate
evaluation on an arbitrary set of values, as in
Fact~\ref{fact:eval-inter}, though a simple observation (see for
example \cite{nz-fmebp-04}) shows that one can evaluate a
bivariate polynomial on $n^2$ points in roughly cubic time by just
extending each set of $n$ points into a grid and using the algorithm
from Fact~\ref{fact:bi:eval-inter-grid}.  However, a stronger result is known

\begin{fact}[General Bivariate Multipoint Evaluation \cite{nz-fmebp-04}]
  \label{fact:bi:general-eval}
  If $A(x,y)$ is a bivariate polynomial of degree at most~$n$ in $x$
  and degree at most $m$ in $y$, in coefficient representation
  ($nm$~coefficients), $A$~can be evaluated simultaneously at
  $N$~different points of $\RR^2$ in $O((N +
  nm)m^{\omega_2/2-1+\eps})$ arithmetic operations, for any $\eps>0$,
  where $\omega_2$ is such that the product of any $n\times n$ and $n
  \times n^2$ matrices can be performed in $O(n^{\omega_2})$ time; it is
  known that $\omega_2<3.334$.
\end{fact}

\section{Definitions and tools from geometry and data structures}
\label{sec:def-ds-geom}

\subsection{Orthogonal range search trees}
\label{sec:ortho}

We say that point $p \in \RR^2$ is \emph{dominated} by point $q \in
\RR^2$ if $x(p)\leq x(q)$ and $y(p)\leq y(q)$; we write $p \leq q$.

\begin{fact}
  Given a set $P=\{p_1,\dots,p_n\}$ of $n$~points in $\RR^2$, a
  \emph{(two-dimensional) orthogonal range tree} is a data structure of size $O(n
  \log n)$ that supports queries of the following type:\footnote{%
    A range tree is actually more powerful, but we will not need all
    of its power here.}  Given a query point $q$, report all the
  points $p\in P$ dominated by $q$, in time $O(k+\log^2 n)$, where $k$
  is the size of the answer.  More specifically, internally the range
  tree is a collection of $O(n \log n)$ \emph{canonical subsets} of $P$, so
  that each point of $P$ lies in $O(\log^2 n)$ canonical subsets, and
  the answer to the query $q$ is represented as a disjoint union of
  $O(\log^2 n)$ such subsets.
\end{fact}

\begin{fact}[Dominating Pair Decomposition]
  \label{fact:pairs-decomposition}
  Given two sets $P, Q \in \RR^2$ each of size $n$, it is possible to
  construct a collection of pairs of subsets $\{(P_i,Q_i)\}$ with the
  following properties:
  \begin{enumerate}
  \item $P_i \subseteq P$, $Q_i \subseteq Q$.
  \item For each pair of points $p \in P$, $q \in Q$ with $p \leq q$,
    there is a unique $i$ so that $p \in P_i$ and $q \in Q_i$.
  \item Each pair in $\bigcup_i (P_i \times Q_i)$ is a dominating pair.
  \item There are at most $O(n\log n)$ pairs of sets in the collection.
  \item $\sum_i (|P_i|+|Q_i|) = O(n \log^2 n)$.
  \item The collection of pairs can be constructed in time $O(n \log^2 n)$.
  \end{enumerate}
\end{fact}

The above construction is not difficult if one starts with a range tree for $P$. The $P_i$ are the canonical sets of the range tree.  Initialize each $Q_i$ to the empty set.   Then execute a query for each $q \in Q$ and add $q$ to the set~$Q_i$ associated with each canonical set $P_i$ that participates in the query.

\subsection{Voronoi diagrams and friends}
\label{sec:voronoi}

Given a set $P$ of $n$ points in the plane, one can partition the
plane into convex polygonal regions, so that each region $V(p)$
consists entirely of points $q$ closer to a given point $p \in P$ than
to any other point of $P$, where ``closer'' is measured with respect to the usual Euclidean distance.  Boundaries of the regions contain points
where the nearest neighbor is non-unique.  The partition is called the
\emph{Voronoi diagram} of $P$ and the sets $V(p)$ are the \emph{Voronoi
cells}.  The Voronoi diagram consists of $n$ cells and a linear number
of edges and vertices separating the cells.

\begin{fact}[Two-Dimensional Voronoi Diagram]
  \label{fact:vor-2d}
  Given an $n$-point set $P \subset \RR^2$, the Voronoi diagram of $P$
  has linear complexity and, in time $O(n \log n)$, can be constructed 
  and preprocessed for logarithmic-time \emph{point location} queries:
  Given a point~$q$, determine the point $p \in P$ whose region~$V(p)$ contains~$q$.
\end{fact}

If the Euclidean distance $|q-p|$ in the above construction is replaced
by $|q-p|/w_p$, where $w_p$ is a positive multiplicative
weight of $p$, we obtain the \emph{multiplicatively weighted Voronoi diagram},
which is similar to the Euclidean version, except that the complexity
can be quadratic in $n$, in the worst case, the edges are circular
arcs, and in general the region of a given Voronoi site can be
disconnected and may contain holes.

\begin{fact}[Two-Dimensional Multiplicatively Weighted Voronoi Diagram \cite{ae-oacwvdp-84}]
  \label{fact:vor-2d-wtd}
  Given an $n$-point set $P \subset \RR^2$ with positive weights
  $w_p$, the multiplicatively weighted Voronoi diagram of $P$ has
  complexity $O(n^2)$ and, in time $O(n^2)$,  can be constructed and
  preprocessed for logarithmic-time \emph{point-location queries}:
  Given a point $q$, determine the point $p \in P$ whose region $V(p)$
  contains $q$.
\end{fact}

The one-dimensional analogs of the Euclidean and multiplicatively
weighted Voronoi diagrams for an $n$-point set have linear complexity
and can be constructed and preprocessed for point location in time
$O(n \log n)$; the weighted version is specifically addressed
in \cite{a-odwvd-86}:

\begin{fact}[One-Dimensional Multiplicatively Weighted Voronoi Diagram]
  \label{fact:vor-1d-wtd}
  Given an $n$-point set $P \subset \RR$ with positive weights
  $w_p$, the multiplicatively weighted Voronoi diagram of $P$ has
  linear complexity and, in time $O(n \log n)$, can be constructed and
  preprocessed for logarithmic-time \emph{point-location queries}:
  Given a point $q \in \RR$, determine the point $p \in P$ whose region $V(p)$
  contains~$q$.
\end{fact}

We will also need the following slightly esoteric variant of the weighted diagram (which of course is a generalization of Fact~\ref{fact:vor-1d-wtd}): 

\begin{fact}[One-Dimensional Slice of a Two-Dimensional Multiplicatively Weighted Voronoi Diagram]
  \label{fact:vor-2d-slice}
  Given an $n$-point set $P \subset \RR^2$ with positive weights $w_p$,
  the multiplicatively weighted Voronoi diagram of $P$ restricted to a
  given line $\ell$ has complexity $O(n)$ and, in time
  $O(n \log n)$, can be constructed and preprocessed for logarithmic-time
  \emph{point-location queries}: Given a point $q$, determine the
  point $p \in P$ whose region $V(p)$ contains $q$.
\end{fact}

\begin{proof}
  Without loss of generality, after a suitable rigid transformation,
  we can assume that the line coincides with the $x$-axis.  The
  squared weighted-distance function $f_p(x)$ from point~$(x,0)$ to site $p$ is
  given by
  \[
    f_p(x) = |p-(x,0)|^2/w^2_p = ((x(p)-x)^2+y(p)^2)/w^2_p,
  \]
  which is a quadratic function of $x$.  As observed, for example, by
  Edelsbrunner and Seidel \cite{es-vda-86}, the minimization diagram
  of these functions coincides with the desired restricted weighted
  Voronoi diagram.  Being quadratic, the graphs of two such functions
  $f_p(x)$ and $f_q(x)$ intersect at most twice and therefore their
  lower envelope and minimization diagram has linear complexity and
  can be constructed by a straightforward $O(n \log n)$ time
  algorithm \cite[Theorem~6.1]{sa-dsstg-95}. 
\end{proof}

We will also need an approximate version of the Voronoi diagram.  We will only state the weighted version of the problem, as that is what we need for our purposes \cite[Theorem~2.16]{wann}:

\begin{fact}[Approximate Two-Dimensional Multiplicatively Weighted Nearest Neighbor \cite{wann}]
  \label{fact:2d-ann-wtd}
  Given an $n$-point set $P \subset \RR^2$ with positive weights
  $w_p$ and a positive number $\eps$, one can preprocess it into a data structure of space $O(n\eps^{-6}\log^4 n)$ in time $O(n\eps^{-6}\log^7 n)$ to support $O(\log n + 1/\eps^3)$  time queries of the form: Given a point $q$, return $p' \in P$ so that $|p'-q|/w_{p'} \leq (1+\eps) |p^*-q|/w_{p^*}$, where $p^*$ is the point in $P$ minimizing $|p-q|/w_p$.

  Alternatively, there is a data structure of space $O(n\eps^{-6}\log^4 n)$  built in time $O(n\eps^{-6}\log^7n)$ with $O(\log(n/\eps))$ query time.
\end{fact}

\end{document}